\documentclass[lettersize,journal]{IEEEtran}

\usepackage{amsmath,amssymb,amsfonts}
\usepackage{amsthm}
\usepackage[linesnumbered, ruled, noend]{algorithm2e}
\usepackage{algorithmic}[noend]
\usepackage{array}
\usepackage[caption=false,font=normalsize,labelfont=sf,textfont=sf]{subfig}
\usepackage{textcomp}
\usepackage{stfloats}
\usepackage{url}
\usepackage{verbatim}
\usepackage{graphicx}
\usepackage{cite}
\usepackage{bm}
\usepackage{color}

\newtheorem{theo}{Theorem}

\newtheorem{defi}{Definition}
\newtheorem{assumption}{Assumption}

\usepackage{lscape}

\begin{document}
\setlength\textfloatsep{8pt}
\setlength\intextsep{8pt}
\title{Bayesian Meta-Learning on Control Barrier Functions with Data from On-Board Sensors}

\author{Wataru Hashimoto, Kazumune Hashimoto, Akifumi Wachi, Xun Shen, Masako Kishida, and Shigemasa Takai
\thanks{{Wataru Hashimoto, Kazumune Hashimoto, Xun Shen, and Shigemasa Takai are with the Graduate School of Engineering, Osaka University, Suita, Japan (e-mail: hashimoto@is.eei.eng.osaka-u.ac.jp, \{hashimoto, shenxun, takai\}@eei.eng.osaka-u.ac.jp). Akifumi Wachi is with LINE Corporation, Tokyo, Japan (email: akifumi.wachi@linecorp.com). Masako Kishida is with the National Institute of Informatics, Tokyo, Japan (email: kishida@nii.ac.jp).
This work is supported by JST CREST JPMJCR201, Japan and by JSPS KAKENHI Grant 21K14184.
}}
}

\markboth{Journal of \LaTeX\ Class Files,~Vol.~14, No.~8, August~2021}%
{Shell \MakeLowercase{\textit{et al.}}: A Sample Article Using IEEEtran.cls for IEEE Journals}


\maketitle

\begin{abstract}
In this paper, we consider a way to safely navigate the robots in unknown environments using measurement data from sensory devices. The control barrier function (CBF) is one of the promising approaches to encode safety requirements of the system and the recent progress on learning-based approaches for CBF realizes online synthesis of CBF-based safe controllers with sensor measurements. However, the existing methods are inefficient in the sense that the trained CBF cannot be generalized to different environments and the re-synthesis of the controller is necessary when changes in the environment occur. Thus, this paper considers a way to learn CBF that can quickly adapt to a new environment with few amount of data by utilizing the currently developed Bayesian meta-learning framework. The proposed scheme realizes efficient online synthesis of the controller as shown in the simulation study and provides probabilistic safety guarantees on the resulting controller. 
\end{abstract}


\section{Introduction}
\IEEEPARstart{E}{NSURING} safety while achieving certain tasks is a fundamental yet challenging problem that has drawn significant attention from researchers in the control and robotics communities over the past few decades. In the control community, such safety requirements are often handled by imposing corresponding state constraints on an optimal control problem. Model Predictive Control (MPC) \cite{MPC} and control synthesis methods based on the certificate functions such as Control Lyapunov Function (CLF) and Control Barrier Function (CBF) \cite{CBF-CLF,CLF,CBF-CLF2} are notable examples of such a method. Desirable theoretical results for these methods such as feasibility, stability, and safety are offered for known system dynamics \cite{MPC,CBF-CLF,CLF,CBF-CLF2} or noisy dynamics with known noise distribution \cite{robust MPC, robust CBF1, robust CBF2}.



The CBF is a prominent tool to impose forward invariance of a safe set in state space. If a dynamical system is control-affine, quadratic programming formulation called CBF-QP \cite{CBF-CLF2} is available, which can produce safe control inputs much faster than many other control methods including MPC, and enables real-time implementations for challenging applications such as autonomous vehicles \cite{autonomous} and bipedal locomotion \cite{bipedal}. 
However, the majority of the existing works presuppose the availability of a valid CBF that represents safe and unsafe regions in state space. This assumption, however, cannot be maintained if a robot is expected to operate autonomously in unknown or uncertain environments. In such situations, it is essential for a robotic system to automatically identify unsafe regions on the fly, employing data from sensory devices. Thus, in this paper, we consider a learning-based method for constructing a CBF based on online measurements obtained from onboard sensors such as LiDAR scanners. 

In the previous studies regarding this line of the topic (e.g., \cite{sample4,lidar1,lidar2}), the robotic system iteratively improves the prediction of the CBF by collecting online sensor measurements without any prior knowledge about the environment or learning task. However, in such cases, even a slight change in the environment necessitates a complete retraining process, which is inefficient and requires a substantial amount of data. To efficiently and effectively learn CBF, adapting it to a new environment with minimal or small amount of training data, we suggest implementing meta-learning \cite{meta1,meta2}, a technique that allows an agent to \textit{learn how to learn}, as opposed to creating it from scratch. During the meta-training process, the model parameters are trained with data collected from various different environments, aiming to learn common underlying knowledge or structure of the learning task. The parameters obtained through meta-learning are then utilized as initial parameters for the online phase in a new environment and updated with online measurements in the environment.
In this paper, we specifically employ a learning scheme combining the recently developed Bayesian meta-learning method \cite{ALPaCA} which is equipped with probabilistic bounds on the prediction \cite{ALPaCA2}, and the Bayesian surface reconstruction method \cite{GPIS1,GPIS3,GPIS4} that can learn the unsafe regions from noisy sensor measurements. 
With this method, we can learn CBF and synthesize a safe controller with few amount of data and provide a probabilistic guarantee on the resulting controller.

\textbf{Related works on learning and CBF
:} Learning-based approaches for certificate functions such as CBF and CLF are one of the active research topics in the area of intersections between control and learning \cite{review}. In the following, we summarize the recent progress in the topic of learning and CBF. The authors of the works \cite{active1,active2,active3,active4,active6, active8} used CBF for active safe learning of uncertain dynamics. The works \cite{active1,active2} consider machine learning approaches to safely reduce the model uncertainty and show the methods yield empirically good performances, while the Gaussian Process (GP) is also used in \cite{active3,active4,active6,active8} and rigorous theoretical results for safety requirements are offered under certain assumptions.
CBF is also used in Imitation Learning (IL) and Reinforcement Learning (RL) \cite{IL1,IL2,RL1} to account for safety concerns. 

While the above studies assume a valid CBF is given and consider leveraging it to impose safety conditions, several previous studies also consider parameterizing CBF itself by Neural Network (NN) and learning it to be a valid barrier function \cite{sample1,demonstration,sample2,sample3,sample4, lidar1,lidar2}. 
The work \cite{sample1} considers jointly learning NNs representing CBF and CLF as well as a control policy based on safe and unsafe samples. The method proposed in \cite{demonstration} considers recovering CBF from expert demonstrations and theoretically guarantees that the learned NN meets the CBF conditions under the Lipsitz continuity assumptions. 
In \cite{sample4}, GP is also used to synthesize CBF.
The works most related to this paper are \cite{lidar1,lidar2} which propose ways to learn CBF with sensor measurements. In \cite{lidar1}, Support Vector Machine (SVM) classifier representing CBF is trained with safe and unsafe samples constructed by using the LiDAR measurements. 
Similarly, the authors of the work \cite{lidar2} present a synthesis method using Signed Distance Function (SDF). 

\textbf{Related works on meta-learning and control
:}
Meta-learning, also known as ``learning to learn," is a  technique capable of learning from previous learning experiences or tasks in order to improve the efficiency of future learning \cite{meta1}. 
Meta-learning has recently become increasingly popular in the control literature. In several previous studies such as \cite{ALPaCA2,metaMPC}, meta-learning is used to learn system dynamics to enable quick adaptation to the changes in surrounding situations. In these works, the trained models are incorporated into MPC which requires relatively heavy online computation.
Meta-learning is also used in policy learning in the context of adaptive control \cite{metaAdapt} and reinforcement learning \cite{meta1,metaRL2}.

\textbf{Contributions:}
Our contributions compared to the existing methods are threefold:
First, compared to the methods that synthesize CBF with sensor measurements \cite{sample4,lidar1,lidar2} which rely on traditional supervised learning, the proposed meta-learning scheme can effectively use past data collected from different environments and produce a prediction of CBF with few amount of data. Consequently, the resulting control scheme realizes less conservative control performance with few amount of data (small number of online CBF updates) as we see in the case study. 
Second, different from the previous NN-based CBF synthesis methods \cite{sample1,sample2,sample3,lidar1,lidar2}, our method can readily take into account uncertainty in data and calculate formal probabilistic bounds on CBF. Based on this, we provide a probabilistic safety guarantee on the proposed control scheme. Although the work \cite{lidar2} takes into account the prediction error in SDF, such an error bound is assumed to be given and actual computation is not provided. 
Third, the learned CBF can readily be incorporated into ordinary QP formulation which can be efficiently solved. Thus, our method can produce control input much faster than the previous methods regarding meta-learning and MPC such as \cite{ALPaCA2,metaMPC}.

\textbf{Notations:}
A continuous functions $\alpha_1:\mathbb{R}_{\geq 0}\rightarrow \mathbb{R}_{\geq 0}$ and $\alpha_2:\ \mathbb{R}\rightarrow \mathbb{R}$ are class $\mathcal{K}$ and extended class $\mathcal{K}_\infty$ function, respectively, if they are strictly increasing with $\alpha_1(0)=0$ and $\alpha_2(0)=0$. $\mathcal{X}_d^2(p)$ is the $p$-th quantile of the $\mathcal{X}^2$ distribution with $d$ degrees of freedom. The maximum and minimum eigenvalues for any positive definite matrix $A$ are defined as $\bar{\lambda}(A)$ and $\underline{\lambda}(A)$, respectively. For vector fields $f$ and a function $F$, $L_f F(x):=\frac{\partial F(x)}{\partial x}f(x)$ is the Lie derivatives of the function $F$ in the direction of $f$.
\section{Problem Statement}\label{problem}

Throughout this paper, we consider the following control affine system: 
\begin{align}\label{system}
    \dot{x} &= f(x)+g(x)u,
\end{align}
where $x\in \mathcal{D} \subset \mathbb{R}^n$ and $u\in \mathcal{U}:=\{ u\in \mathbb{R}^m\mid  Au \leq b \}$ with $A\in \mathbb{R}^{m\times m}$ and $b\in \mathbb{R}^m$ are the admissible system states and control inputs respectively. The functions $f: \mathbb{R}^n  \rightarrow \mathbb{R}^n$ and $g: \mathbb{R}^n \rightarrow \mathbb{R}^{n\times m}$ are continuously differentiable functions representing dynamics of a robotic system (e.g., ground vehicles or drones) equipped with sensors such as LiDAR. 
The state space is assumed to contain \textit{initially unknown} $N_{\mathrm{obs}}$ obstacles and the \textit{outside} of each obstacle is defined by the following:
\begin{align}\label{cbf}
    \mathcal{C}_{\xi_i} = \{ x\in \mathcal{D}\subset \mathbb{R}^n\mid h(x;\xi_i)\geq 0 \},\ i \in \mathbb{N}_{1: N_{\mathrm{obs}}},
\end{align}
where $h(\cdot;\xi):\mathbb{R}^n\rightarrow \mathbb{R}$ is a continuously differentiable function with \textit{unobserved} latent variable $\xi$ that encodes the geometric information about an obstacle, e.g., the position and shape of an obstacle. We assume realizations of the variable, $\xi_i$, $i\in\mathbb{N}_{1: N_{\mathrm{obs}}}$ are determined by samples from a probabilistic distribution $p(\xi)$ and they are fixed during the control execution. For simplicity, the obstacles are chosen not to overlap with each other. 

Our goal is to synthesize a controller that achieves a given task (e.g., goal-reaching task) without deviating from the safe region $\mathcal{C}=\bigcap_{i=1}^{N_{\mathrm{obs}}}\mathcal{C}_{\xi_i}$ given initial state $x(0)\in \mathcal{C}$. 
Since the safe/unsafe regions are initially unknown, the robotic system needs to identify them based on the sensor measurements. 
A naive approach toward this end is to use supervised learning based on measurement data obtained under a pre-determined environment. 
However, in this case, the trained model cannot be generalized to different environments, and even slight changes in the environment necessitate a complete re-fitting of the model, which is not only inefficient but also demands a substantial amount of data.
Thus, this paper considers a way to effectively use past data regarding the different environments sampled from the distribution $p(\xi)$ to quickly learn safe/unsafe regions in a new environment and synthesize the safe controller with few amount of online data. 

In the following, we first introduce the control notions suited for our purpose called Control Barrier Function (CBF) and Control Lyapunov Function (CLF) in Section \ref{sec:CBF}. Then, the meta-learning procedures of CBF and a way to synthesize QP with the learned CBF are discussed in Section \ref{sec:meta} and \ref{sec:control}.
\section{CBF-CLF-QP}\label{sec:CBF}

In this section, we first introduce the Zeroing Control Barrier Function (ZCBF) as a tool to enforce safety on the system (\ref{system}).
The following discussion is a summary of the works \cite{CBF-CLF,CBF-CLF2}. In ZCBF, the safety is defined as the forward-invariance of a safe set $\mathcal{C}$ which is defined by the super zero level set of a function $h: \mathbb{R}^n\rightarrow \mathbb{R}$, i.e., the system (\ref{system}) is considered to be safe if $x(t)\in \mathcal{C}:=\{ x \mid h(x)\geq 0 \}$ holds for all $t\geq 0$ when $x(0)\in \mathcal{C}$.
Given an extended class $\mathcal{K}_\infty$ function $\alpha_C: \mathbb{R}\rightarrow \mathbb{R}$, the ZCBF is defined as follows.
\begin{defi}
Let $\mathcal{C} \subset \mathcal{D} \subset \mathbb{R}^n$ be the zero super level set of a continuously differentiable function $h:\ \mathbb{R}^n\rightarrow \mathbb{R}$. Then, the function $h$ is a zeroing control barrier function (ZCBF) for (\ref{system}) on $\mathcal{C}$ if there exists an extended class $\mathcal{K}_\infty$ function $\alpha_C$ such that, for all $x\in \mathcal{D}$ the following inequality holds.
\begin{align}\label{cbf-condition}
    \sup_{u\in \mathbb{R}^m} \left\{ L_f h(x) + L_g h(x)u + \alpha_C(h(x)) \right\}\geq 0.
\end{align}
We additionally define the set of control inputs that satisfy the CBF condition as follows:
\begin{align}
U_C(x) = \{ u\in \mathbb{R}^m\mid L_f h(x) + L_g h(x)u + \alpha_C(h(x)) \geq 0\}.\notag
\end{align}

\end{defi}
If the function $h$ is a valid CBF for (\ref{system}) on $\mathcal{C}$, the safety of the system (\ref{system}) is guaranteed as the following theorem.

\begin{theo}[\cite{CBF-CLF2}]\label{theo:cbf}
Let $\mathcal{C}\subset \mathcal{D}\subset \mathbb{R}^n$ be a super zero level set of a continuously differentiable function $h:\ \mathbb{R}^n\rightarrow \mathbb{R}$ and the function $h$ is ZCBF for (\ref{system}). Then, any policy that selects control inputs from $U_C(x)$ renders the safe set $\mathcal{C}$ forward invariant for the system (\ref{system}). 
\end{theo}

In this paper, the control objectives, except for the safety requirements, are encoded through the control Lyapunov function (CLF) which is given in advance. With a class $\mathcal{K}$ function $\alpha_V$, the definition of CLF is formally given as follows \cite{CBF-CLF}. 
\begin{defi}
A continuously differentiable function $V:\ \mathbb{R}^n\rightarrow \mathbb{R}_{\geq 0}$ is a CLF if $V(x_e)=0$ with equilibrium point $x_e\in \mathcal{D}$ and there exists a class $\mathcal{K}$ function $\alpha_V: \mathbb{R}\rightarrow \mathbb{R}$ such that, for all $x\in \mathcal{D}\subset \mathbb{R}^n$ the following inequality holds.
\begin{align}\label{clf-condition}
    \inf_{u\in \mathbb{R}^m} \left\{ L_f V(x) + L_g V(x)u + \alpha_V(V(x)) \right\}\leq 0.
\end{align}
\end{defi}
Given a CLF $V$, we consider the set of all control inputs that satisfy (\ref{clf-condition}) for a state $x\in \mathcal{D}$ as
\begin{align}
U_V(x) = \{ u\in \mathbb{R}^m\mid L_f V(x) + L_g V(x)u + \alpha_V(V(x)) \leq 0\}.\notag
\end{align}
If a function $V$ is CLF, the stability of the equilibrium point $x_e$ is guaranteed as the following theorem.
\begin{theo}[\cite{CLF}]
Let a continuously differentiable function $V:\ \mathbb{R}^n\rightarrow \mathbb{R}_{\geq 0}$ be CLF for (\ref{system}). Then, any policy that selects control inputs from $U_V(x)$ renders the equilibrium point $x_e$ of the system (\ref{system}) asymptotically stable. 
\end{theo}

Having CBF and CLF defined above, we can synthesize a controller satisfying CBF and CLF conditions (\ref{cbf-condition}) and (\ref{clf-condition}) with quadratic programming (QP) as follows: 
\begin{subequations}\label{qp}
\begin{align}\
    &\min_{(u,\epsilon)\in \mathbb{R}^{m+1}}\  \frac{1}{2} u^\top H u+\lambda \epsilon^2, \\
    & \mathrm{s.t.} \ L_fV(x)+L_gV(x)u+\alpha_V(V(x))\leq \epsilon, \label{qp clf}\\
     & \ \ \ \ \ L_fh(x)+L_gh(x)u+\alpha_C(h(x))\geq 0 \label{qp cbf} \\
     & \ \ \ \ \ Au \leq b,
\end{align}
\end{subequations}
where $H\in \mathbb{R}^{m\times m}$ is a positive definite matrix and $\lambda\ (>0)$ is a tunable parameter for relaxing the CLF condition (\ref{qp clf}). The relaxation of the CLF condition ensures the solvability of the QP. 
If the safety is defined by multiple CBFs, we can handle them by adding the corresponding CBF constraints to the optimization problem.
\section{Bayesian Meta-Learning of CBF}\label{sec:meta}
In this section, we consider learning the safe region (\ref{cbf}) based on the measurement data from sensor devices. 
To this end, we employ a technique based on the Gaussian process implicit surfaces (GPIS) \cite{GPIS1,GPIS3,GPIS4} which is developed for non-parametric probabilistic reconstruction of object surfaces. 
Note that we consider the case of $N_{\mathrm{obs}}=1$ for a while, and the case of multiple obstacles is discussed later in the next section. In this paper, the implicit surfaces (IS) are defined by the distance from the surface of objects to a certain point $z$ in 2D or 3D space, and the inside and outside of the object surfaces are interpreted through the sign of the function as follows. 
\begin{align}\label{IS}
    h_{\mathrm{IS}}(z;\xi) = \left\{
\begin{array}{ll}
d_\xi(z) & \text{if $z$ is outside the obstacle}\\
0 & \text{if $z$ is on the surface}\\
-d_\xi(z) & \text{if $z$ is inside the obstacle}
\end{array}
\right.
\end{align}
where $z\in \mathbb{R}^b$, $b\in \{ 2,3\}$ is a position in a 2D or 3D space and $d_\xi:\mathbb{R}^b\rightarrow \mathbb{R}_{\geq 0}$ is the function that returns the minimum Euclidean distance between $z$ and surface of the obstacle corresponding to $\xi \sim p(\xi)$. Though this paper focuses on a 2D environment, our method can easily be extended to a 3D case.
In GPIS, the function $h_{\mathrm{IS}}$ is estimated with the Gaussian process (GP), which enables us to deal with uncertainty arising from measurement noise and scarcity of data. The GP is estimated with a dataset $\{ (z^i,h_{\mathrm{IS}}(z^i,\xi)+ \varsigma_i)\}_{i=1}^{n_\mathrm{data}}$, where $ \varsigma_i$ is bounded $\sigma$-subgaussian noise and $n_\mathrm{data}$ is the number of data.

In the following, we explain how to construct a dataset and apply Bayesian meta-learning to the training of $h_{\mathrm{IS}}$.
\subsection{Dataset Construction}\label{dataset}
\begin{figure}[tb]
 \begin{center}
  \includegraphics[width=1\hsize]{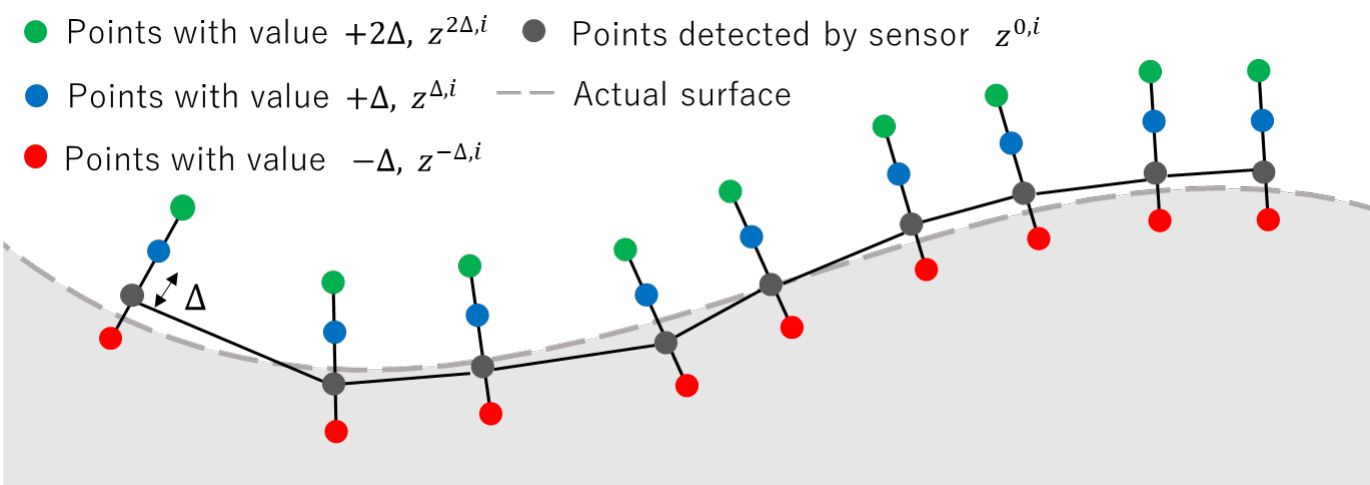}
 \end{center}
 \caption{The construction of data points from noisy sensor measurements: The gray dots are the points on the surface detected by the sensors. The green, blue, and red dots represent the points that have distance $2\Delta$, $\Delta$, $-\Delta$ from the surface (negative value means inside of the obstacle) which are calculated by using surface normal approximation explained in Section \ref{dataset} (we show the case of $n_-=1$ and $n_+=2$).}
 \label{fig:data}
\end{figure}
In this subsection, we explain how to construct a dataset to learn implicit surfaces (\ref{IS}) based on the measurement data from sensors such as LiDAR. A pictorial image for obtaining training data from noisy sensor measurements is shown in Fig. \ref{fig:data}. 
First, assume we have a point cloud on the surface $\mathcal{P}^0= \cup_{i=1}^{n_{\mathrm{surf}}}\{ z^{0,i} \}$, where $n_{\mathrm{surf}}\in \mathbb{Z}_{>0}$ is the number of surface points and $z^{0,i}\in \mathbb{R}^2$, $\forall i\in \mathbb{N}_{1:n_{\mathrm{surf}}}$ are the world Cartesian coordinates of surface points computed with the raw depth readings produced by sensors and current position of the robot. Then, for each surface point $z^{0,i}$, the normal to the surface of the obstacle at that point is approximated by the perpendicular to the segment between the point and the nearest neighbor point within the same scan \cite{GPIS4}. Under the assumption that the sensor measurements on the surface are dense enough, this method yields accurate approximations. 
For each surface point $z^{0,i}$ and corresponding approximated normal to the surface, we construct a set of points $\cup_{p=-n_-}^{n_+}\{ z^{p\Delta, i} \}$ with $\Delta \in \mathbb{R}_{>0}$ and $n_+,n_-\in\mathbb{Z}_{\geq0}$, where $z^{p\Delta, i}$, $\forall p\in \mathbb{Z}_{-n_-:n_+}$ are the representations of the points that are $p\Delta$ away from the surface of the obstacle (see Fig \ref{fig:data}). A negative $p$ value means that the point $z^{p\Delta, i}$ is inside the obstacle.
Then, we construct the data set to train $h_{\mathrm{IS}}$ as $D=\cup_{i=1}^{n_{\mathrm{surf}}}\cup_{p=-n_-}^{n_+}\{(z^{p\Delta,i},p\Delta)\}$. 
\subsection{Bayesian meta-learning of the function $h_{\mathrm{IS}}$}
\begin{figure*}[tb]
 \begin{center}
  \includegraphics[width=0.95\hsize]{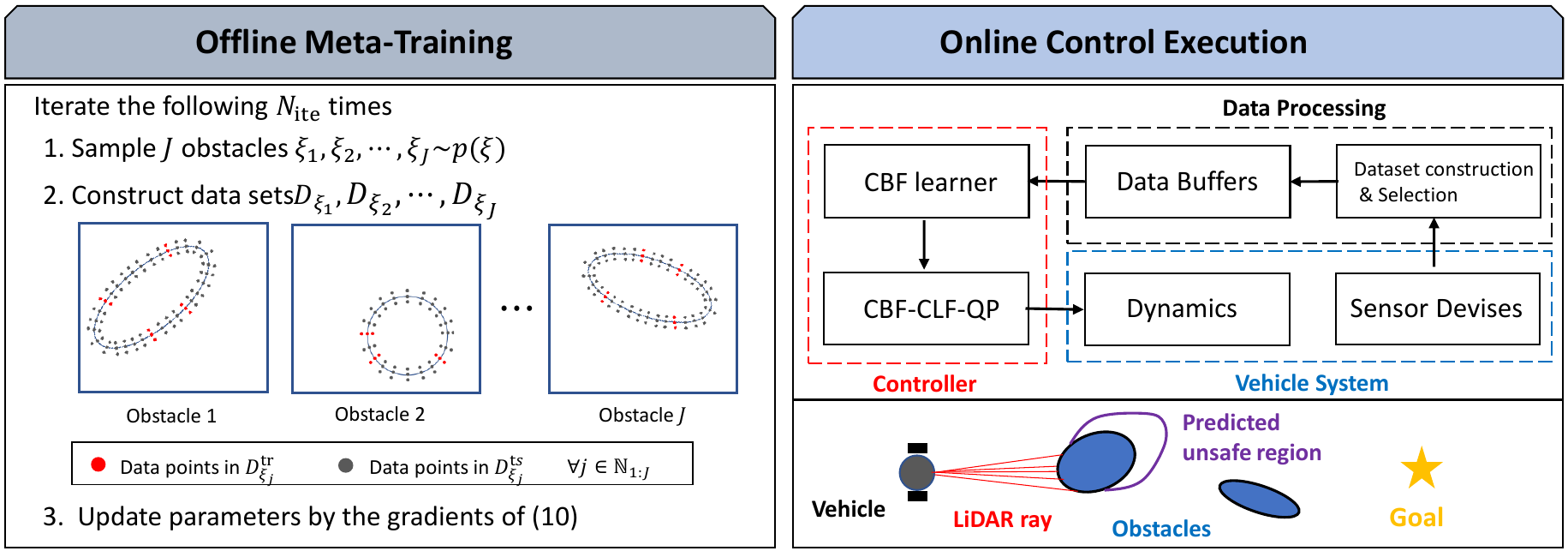}
 \end{center}
 \caption{The summation of the proposed approach: The proposed method consists of \textit{offline meta-training} and \textit{online control execution}. The offline training follows the procedures in Section \ref{offline}. In each iteration, we first sample $J$ obstacles from the distribution $p(\xi)$ and construct the datasets $D_{\xi_j}^{\mathrm{tr}}$ and $D_{\xi_j}^{\mathrm{ts}}$ for each of these obstacles. Then, the loss (\ref{loss}) of the prediction made by the model adopted from $D_{\xi_j}^{\mathrm{tr}}$ is evaluated for test data $D_{\xi_j}^{\mathrm{ts}}$ and the parameters are updated with the gradients of the loss. 
 In the online control execution, the control input is generated by CBF-CLF-QP (\ref{qp}) with the CBFs obtained through the adaptation of the posterior parameters (\ref{post}) and the probabilistic bounds (\ref{cbf pred}). The dataset for the adaptation is constructed using the sensor measurements obtained during the operation by following procedures in Section \ref{dataset} and the data selection scheme discussed in \ref{criteria}.}
 \label{fig:proposed}
\end{figure*}
When we train the IS function online, the efficiency of the training may be problematic while the model trained offline is not capable of generalizing to new environments. 
Thus, instead of relying on either online or offline training, we consider employing a meta-learning \cite{meta1} scheme, which enables quick prediction on $h_{\mathrm{IS}}$ with a few amount of online data by effectively using data from different settings. 
In the following, we first introduce a Bayesian meta-learning framework ALPaCA (Adaptive Learning for Probabilistic Connectionist Architectures) \cite{ALPaCA}. Then, the probabilistic bounds on the prediction are subsequently discussed. 
\subsubsection{Overview of ALPaCA}
We first parameterize the prediction of the function $h_{\mathrm{IS}}$ by the following form:
\begin{align}\label{model}
    \hat{h}_{\mathrm{IS}}(z;\theta) = \theta^\top \phi_w(z),
\end{align}
where $\phi_w: \mathbb{R}^2 \rightarrow  \mathbb{R}^d$ represents a feed-forward neural network with parameter $w$ and $\theta \in \mathbb{R}^d$ is a coefficient matrix that follows a Gaussian distribution $\mathcal{N}(\bar{\theta},\sigma^2\Lambda^{-1})$ which encodes the information and uncertainty associated with the unknown variable $\xi\sim p(\xi)$. Here, $\bar{\theta}\in \mathbb{R}^d$ denotes the mean parameters and $\Lambda \in \mathbb{R}^{d\times d}$ is the positive definite precision matrix. Once the parameters of $\phi_w$ and the prior distribution, ($w,\bar{\theta}_0,\Lambda_0$), are determined through the offline meta-training procedure explained later, and the new measurement data along with a fixed $\xi\sim p(\xi)$ is given as $D_\xi = \{ (z_i, y_i)\}_{i=1}^{n_{\mathrm{adapt}}}$, the posterior distribution on $\theta$ is obtained as follows.
\begin{align}\label{post}
    \Lambda_\xi = \Phi_\xi^\top \Phi_\xi + \Lambda_0,\ \ 
    \bar{\theta}_\xi = \Lambda_\xi^{-1}(\Phi_\xi^\top G_\xi+ \Lambda_0 \theta_0),
\end{align}
where $G_\xi^\top = [y_1,y_2,\ldots, y_{n_{\mathrm{adapt}}}]\in \mathbb{R}^{n_{\mathrm{adapt}}}$, and $\Phi_\xi^\top=[\phi_w(z_0),\phi_w(z_1),\ldots, \phi_w(z_{n_{\mathrm{adapt}}})]\in \mathbb{R}^{d\times n_{\mathrm{adapt}}}$. 
The above computation is based on Bayesian linear regression. Here, the mean and precision matrices are subscripted by $\xi$ to explicitly show that they are calculated using the data associated with unknown variable $\xi\sim p(\xi)$.
Then, for any given $z \in \mathbb{R}^2$, the mean and variance of the posterior predictive distribution are obtained as
\begin{align}\label{post2}
    \mu_\xi(z) = \bar{\theta}_\xi^\top \phi_w(z),\ \Sigma_\xi(z) = \sigma^2(1+\phi_w(z)^\top\Lambda_\xi^{-1}\phi_w(z)).
\end{align}

\subsubsection{Offline meta-training procedures}\label{offline}
In the offline meta-training procedure, the parameters ($w$, $\bar{\theta}_0$, $\Lambda_0$) are trained so that the previously discussed online adaptation yields accurate prediction (The pseudo-code of the offline training is summarized in Algorithm \ref{alg2} in Appendix \ref{pseudo}).
To this end, we iterate the following procedures prescribed $N_{\mathrm{ite}}\in \mathbb{Z}_{>0}$ times.
First, we prepare $J\in \mathbb{Z}_{>0}$ data sets by sampling obstacles $\xi_1,\xi_2,\ldots,\xi_J$ from the distribution $p(\xi)$ and collect the data for each obstacle as $D_{\xi_j}=\cup_{i=1}^{n_j}\{(z_i^j, y_i^j)\}$, where $n_j$ is the number of data within $D_{\xi_j}$. Then, we randomly split each dataset $D_{\xi_j}$ into training set $D_{\xi_j}^{\mathrm{tr}}=\cup_{i=1}^{n_{j}^{\mathrm{tr}}}\{(z_i^{\mathrm{tr},j}, y_i^{\mathrm{tr},j})\}$ and test set $D_{\xi_j}^{\mathrm{ts}}=\cup_{i=1}^{n_{j}^{\mathrm{ts}}}\{(z_i^{\mathrm{ts},j}, y_p^{\mathrm{ts},j})\}$, where $n_{j}^{\mathrm{tr}}$ and $n_{j}^{\mathrm{ts}}$ are number of data within $D_{\xi_j}^{\mathrm{tr}}$ and $D_{\xi_j}^{\mathrm{ts}}$ with $n_{j}^{\mathrm{tr}}+n_{j}^{\mathrm{ts}}=n_j$. In each iteration, $n_{j}^{\mathrm{tr}}$ is also randomly chosen from the uniform distribution over $\{1,\ldots,n_j \}$.
The data within $D_{\xi_j}^{\mathrm{tr}}$ is used to calculate the posterior distribution on the parameter $\theta_{\xi_j}$ (i.e., mean $\bar{\theta}_{\xi_j}$ and precision matrix $\Lambda_{\xi_j}$) through (\ref{post}) while data within $D_{\xi_j}^{\mathrm{ts}}$ is used to evaluate the accuracy of the model (\ref{model}) with the posterior parameters $\bar{\theta}_{\xi_j}$ and $\Lambda_{\xi_j}$ calculated above. We define the meta-learning objective by the following marginal log-likelihood across the data and update the parameters through stochastic gradient descent.
\begin{align}\label{loss}
    \ell(&\bar{\theta}_0,\Lambda_0,w) := \sum_{j=1}^J \sum_{i=1}^{n_j^{\mathrm{ts}}}\log p(y_i^{\mathrm{ts},j}\mid z_i^{\mathrm{ts},j})\notag \\
    & = \sum_{j=1}^J \sum_{i=1}^{n_j^{\mathrm{ts}}}(\log(1+\phi_w(z_i^{\mathrm{ts},j})^\top \Lambda_{\xi_j} \phi_w(z_i^{\mathrm{ts},j}))\notag \\
   &\ \ \quad+ (y_{i}^{\mathrm{ts},j}-\bar{\theta}_{\xi_j}^\top\phi_w(z_i^\mathrm{ts,j}))^\top \Sigma_{\xi_j}(z_i^\mathrm{ts,j}) \notag \\&\qquad \quad (y_{i}^{\mathrm{ts},j}-\bar{\theta}_{\xi_j}^\top\phi_w(z_i^\mathrm{ts,j}))).
\end{align}

\subsubsection{Probabilistic bounds on CBF}\label{sec bound}
From the discussions above, we can compute the predictive distribution of the function value $h_{\mathrm{IS}}(z;\xi)$ for any $z\in \mathbb{R}^2$ and $\xi\sim p(\xi)$ by (\ref{post2}) once the offline meta-training procedure discussed in Section \ref{offline} has been done and new measurements are obtained along with the online execution. Here, we consider deriving a deterministic function with probabilistic guarantees, which can readily be incorporated in the QP formulation (\ref{qp}). The following discussion follows from \cite{ALPaCA2}. Before deriving the probabilistic bounds on the prediction of $h_{\mathrm{IS}}$, we make the following two assumptions on the quality of the offline meta-training, which are realistic and intuitive ones. 
\begin{assumption}\label{assumption1}
For all $\xi\sim p(\xi)$, there exists $\theta^*_\xi\in \mathbb{R}^d$ such that
\begin{align}
    {\theta^*_\xi}^\top \phi_w(z) = h_{\mathrm{IS}}(z;\xi),\ \ \forall z\in \mathbb{R}^2.
\end{align}
\end{assumption}
\begin{assumption}\label{assumption2}
    For all $\xi \sim p(\xi)$, we have the following:
    \begin{align}
        \mathbb{P}(\| \theta^*_\xi-\bar{\theta}_0 \|_{\Lambda_i}^2\leq \sigma^2\mathcal{X}_d^2(1-\delta))\geq 1-\delta.
    \end{align}
\end{assumption}
Assumption \ref{assumption1} implies that the meta-learning model (\ref{model}) is capable of fitting the true function while Assumption \ref{assumption2} says that the uncertainty in the prior is conservative enough. 
Under these assumptions, we can derive the bounds on the function $h_{\mathrm{IS}}$ as follows.
\begin{theo}\label{theo:bound}
    Suppose the offline meta-training of the parameters $w$, $\bar{\theta}_0$, and $\Lambda_0$ is done to satisfy Assumption \ref{assumption1} and \ref{assumption2}, and the parameters defining posterior distribution $\bar{\theta}_\xi$ and $\Lambda_\xi$ are computed by (\ref{post}) with data from the new environment $\xi\sim p(\xi)$. Then, the L-1 norm of the difference between the true function value $h_{\mathrm{IS}}(z;\xi)$ and the mean prediction $\bar{h}_{\mathrm{IS}}(z;\theta_\xi)=\bar{\theta}_\xi\phi(z)$ with $\theta_\xi\sim \mathcal{N}(\bar{\theta}_\xi,\Lambda_\xi)$, for all $z\in \mathbb{R}^2$ and $\xi \sim p(\xi)$ is bounded as the following. 
    \begin{align}\label{bound}
      |h_{\mathrm{IS}}(z;\xi)-\bar{h}_{\mathrm{IS}}(z;\theta_\xi)|\leq  |\phi(z)^\top (\bar{\theta}_\xi-\theta^*_\xi)|\leq \|\phi(z) \|_{\Lambda_\xi^{-1}} \beta_\xi,
    \end{align}
    with
    \begin{align}
        \beta_\xi = \sigma\left( \sqrt{2\log \left( \frac{1}{\delta}\frac{\mathrm{det}(\Lambda_\xi)^{1/2}}{\mathrm{det}(\Lambda_0)^{1/2}}\right)} + \sqrt{\frac{\bar{\lambda}(\Lambda_0)}{\underline{\lambda}(\Lambda_\xi)}\mathcal{X}_d^2(1-\delta)}\right),\notag
    \end{align}
    with probability at least $(1-2\delta)$.
\end{theo}
\begin{proof}
Proof of this result follows from \cite{ALPaCA2} and \cite{bandit}. In the middle of the proof of Theorem 1 in \cite{ALPaCA2}, it is shown that $|a^\top (\bar{\theta}_\xi-\theta^*_\xi)|\leq \|a \|_{\Lambda_\xi^{-1}} \beta_\xi$ holds for any $a \in \mathbb{R}^d$ with probability at least $(1-2\delta)$ (see pp. 18 of \cite{ALPaCA2}). 
Then, (\ref{bound}) is obtained by substituting $\phi_w(z)$ to $a$ in the inequality.
\end{proof}
Since the difference between the true value and mean prediction is bounded by $\|a \|_{\Lambda_\xi^{-1}} \beta_\xi$, we can provide the lower bound of $h_{\mathrm{IS}}$ corresponding to a robot state $x$ as the following.
\begin{align}\label{cbf pred}
    h^b(x;\theta_\xi)
    := \bar{h}_{\mathrm{IS}}(v(x);\theta_\xi)-\|\phi(v(x)) \|_{\Lambda_\xi^{-1}} \beta_\xi,
\end{align}
where $v: \mathbb{R}^n\rightarrow \mathbb{R}^2$ is the function that maps the robot's state to the robot's position in 2D space. By using $h^b$ as the CBF, we can provide a probabilistic guarantee of the resulting controller as elaborated in Section \ref{sec:control}.
\section{Online Control Execution}\label{sec:control}

\begin{algorithm}[t]\label{alg}
{\small
\SetKwInOut{Input}{Input}
\SetKwInOut{Output}{Output}
\Input{$f$, $g$ (system dynamics model); $p(\xi)$ (distribution of obstacles); $(w,\bar{\theta}_0,\Lambda_0)$ (prior parameters); $\delta$ (confidence level); $V$ (CLF); $T$ (execution time)}
$\xi_{i}\sim p(\xi),\ \mathcal{B}_{i}\leftarrow{\emptyset},\ \forall i\in \mathbb{N}_{1:N_{\mathrm{obs}}}$; \\

\While{$t<T$}{
\If{$t=k\Delta_{\mathrm{lidar}}$, $k\in \mathbb{Z}_{\geq 0}$}{
\textbf{[CBF update]}\\
\For{$i=1:N_{\mathrm{obs}}$}{
Obtain sensor measurements $\mathcal{P}^0_i$;\\
Construct dataset $D_i$ by procedures in Sec \ref{dataset}\\
Update the buffer $\mathcal{B}_i$ by procedures in Sec \ref{criteria};\\

\If{$\mathcal{B}_i\neq \emptyset$}{
Compute the posterior parameters $\bar{\theta}_{\xi_i}$ and $\Lambda_{\xi_i}$ by (\ref{post}); \\
Set $h^b(\cdot;\theta_{\xi_i}) = {h}^{b}_{IS}(v(\cdot);\theta_{\xi_i})$ by (\ref{cbf pred});\\}}}
\textbf{[Control execution]}\\
Solve the QP (\ref{qp}) with the CBFs obtained above;\\
Solve (\ref{system}) and update state $x(t)$;}

    \caption{Online control execution} 
    }
\end{algorithm}
Given system (\ref{system}), meta-learned parameters $(w,\bar{\theta}_0,\Lambda_0)$, and CLF $V$ that encodes a control objective, the procedures in the online control execution is as summarized in Algorithm \ref{alg}. 
Before the execution, an environment is determined by sampling the obstacles $\xi_i,\ \forall i\in \mathbb{N}_{1:N_\mathrm{obs}}$ from the distribution $p(\xi)$ and the data buffers $\mathcal{B}_i$ for $\xi_i,\  \forall i\in \mathbb{N}_{1:N_\mathrm{obs}}$ are initialized by empty sets (Line 1). Then, the control is executed by repeatedly solving the QP (\ref{qp}) with the CBFs (\ref{cbf pred}) defined using the current posteriors $\bar{\theta}_{\xi_i}$ and $\Lambda_{\xi_i}$ calculated by the data in the current buffers (Line 13,14). The buffers $\mathcal{B}_i$ and the posterior parameters $\bar{\theta}_{\xi_i}$ and $\Lambda_{\xi_i}$ $\forall i\in \mathbb{N}_{1:N_{\mathrm{obs}}}$ are updated at every $\Delta_{\mathrm{lidar}}\in \mathbb{R}_{>0}$ [s] through the procedures in the following subsection (Line 3-11). Since each buffer $\mathcal{B}_i$ is empty at the beginning of the execution, we ignore the CBF constraint for $\xi_i$ until surface points of the corresponding obstacle are detected. To justify this, we make the following assumption.
\begin{assumption}\label{assumption3}
The length of the laser is long enough and the measurements are taken constant enough such that the control invariant of the safe region is maintained even if the unobserved obstacles are ignored in the current QP.
\end{assumption}

In the following subsections, we discuss the update scheme of each buffer $\mathcal{B}_i$ and the theoretical results of the proposed control scheme.
\subsection{The update scheme of each buffer $\mathcal{B}_i$}\label{criteria}
The update procedures of each buffer $\mathcal{B}_i$ are discussed here. We ignore the subscription $i$ in the following discussion for simplicity of notation.
At time instance $k$, suppose we have new data points $D_k=\cup_{i=1}^{n_k}\cup_{p=n_-}^{n_+}\{(z_k^{+p\Delta,i},p\Delta)\}$ obtained by following the procedures in Section \ref{dataset} and the buffer $\mathcal{B}$ which is constructed by the data points obtained at past time instances 0 to $k-1$.
Intuitively, the data points in $D_k$ will include many data points that are similar to those obtained in the past instances if the time interval of updating the dataset is small. Thus, it is important to select only the informative data points before adding them to the buffer. Here, we consider the data selection scheme for this purpose based on \cite{select}. We repeat the following procedures for $i\in \mathbb{N}_{1:n_k}$. First, we compute the predictive variance at a data point on the surface $z_k^{0,i}$ with posterior parameters calculated by the current $\mathcal{B}$ as $\Sigma_{\xi}(z_k^{0,i})$.
Then, if the value $\Sigma_{\xi}(z_k^{0,i})$ is larger than the prescribed threshold $\eta\in \mathbb{R}_{\geq 0}$, the data points $\cup_{p=n_-}^{n_+}(z_k^{p\Delta,i},p\Delta)$ are considered to be informative for the current model and thus added to the buffer $\mathcal{B}$. 
\subsection{Theoretical Results}

The probabilistic safety guarantees for our control scheme are summarized in the following theorem. 
\begin{theo}
   Suppose Assumption \ref{assumption3} holds and the offline meta-training of the parameters $w$, $\bar{\theta}_0$, and $\Lambda_0$ is performed to satisfy Assumption \ref{assumption1} and \ref{assumption2}, and feasible solution of the QP in Algorithm \ref{alg} is obtained at all time $t<T$. Then, $x(t)\in \mathcal{C}_{\xi_i}$ holds with probability at least $1-2\delta$ for all $t\in \mathbb{R}_{>0}$ and $i\in \mathbb{N}_{1:N_{\mathrm{obs}}}$.
\end{theo}
\begin{proof}
    From Theorem \ref{theo:cbf}, the implicit control policy defined by the QP in Algorithm \ref{alg} keeps the system states within $\hat{\mathcal{C}}_{\xi_i}:=\{ x\in \mathcal{D}\subset \mathbb{R}^n\mid h^b(x,\xi_i)\geq 0 \}$ for all of the observed obstacles $\xi_i$. Since each state in $\hat{\mathcal{C}}_{\xi_i}$ is included in the actual safe set $\mathcal{C}_{\xi_i}$ with probability at least $1-2\delta$ from Theorem \ref{theo:bound}, we can guarantee $x(t)\in \mathcal{C}_{\xi_i}$ with probability at least $1-2\delta$ for all $t\in \mathbb{R}_{>0}$ and the observed obstacles $\xi_i$. This result and Assumption \ref{assumption3} prove the theorem.
\end{proof}

\section{Case Study}\label{case study}
We test the proposed approach through a simulation of vehicle navigation in a 2D space. 
All the experiments are conducted on Python running on a Windows~10 with a 2.80 GHz Core i7 CPU and 32 GB of RAM. TensorFlow and PyBullet \cite{pybullet} are used for the meta-learning of the IS function and implementation of the vehicle system with LiDAR scanners, respectively.
We define the controlled system by the unmanned vehicle that has the following dynamics:
\begin{align}\label{car}
    \underbrace{\begin{bmatrix}
        \dot{q_x} \\
        \dot{q_y} \\
        \dot{\vartheta}
    \end{bmatrix}}_{\dot{x}}
    =
    \underbrace{\begin{bmatrix}
        0\\
        0\\
        0
    \end{bmatrix}}_{f(x)}
    +
    \underbrace{\begin{bmatrix}
        \cos \vartheta & -\ell \sin \vartheta \\
        \sin \vartheta & \ell \cos \vartheta\\
        0 & 1
    \end{bmatrix}}_{g(x)}
    \underbrace{\begin{bmatrix}
        v\\
        \omega
    \end{bmatrix}}_{u},
\end{align}
where $q_x$, $q_y$, and $\vartheta$ are $x-y$ coordinates of the robot position and heading angle of the robot, $v$ and $\omega$ are the velocity and angular velocity of the vehicle, respectively. The control inputs and system states are defined as $u = [v,\omega]^\top\in \mathbb{R}^2$ and $x=[q_x,q_y,\vartheta] \in \mathbb{R}^3$, respectively. 
Note that in (\ref{car}), we consider the dynamics of a point off the axis of the robot by a distance $\ell\ (>0)$ to make the system relative degree one \cite{car dynamics}.
Moreover, the LiDAR scanner which has a 360-degree field of view and radiates 150 rays with a 3-meter length is assumed to be mounted on the vehicle. 
\subsection{Meta learning of unsafe regions}
\begin{figure}[tb]
 \begin{center}
  \includegraphics[width=1\hsize]{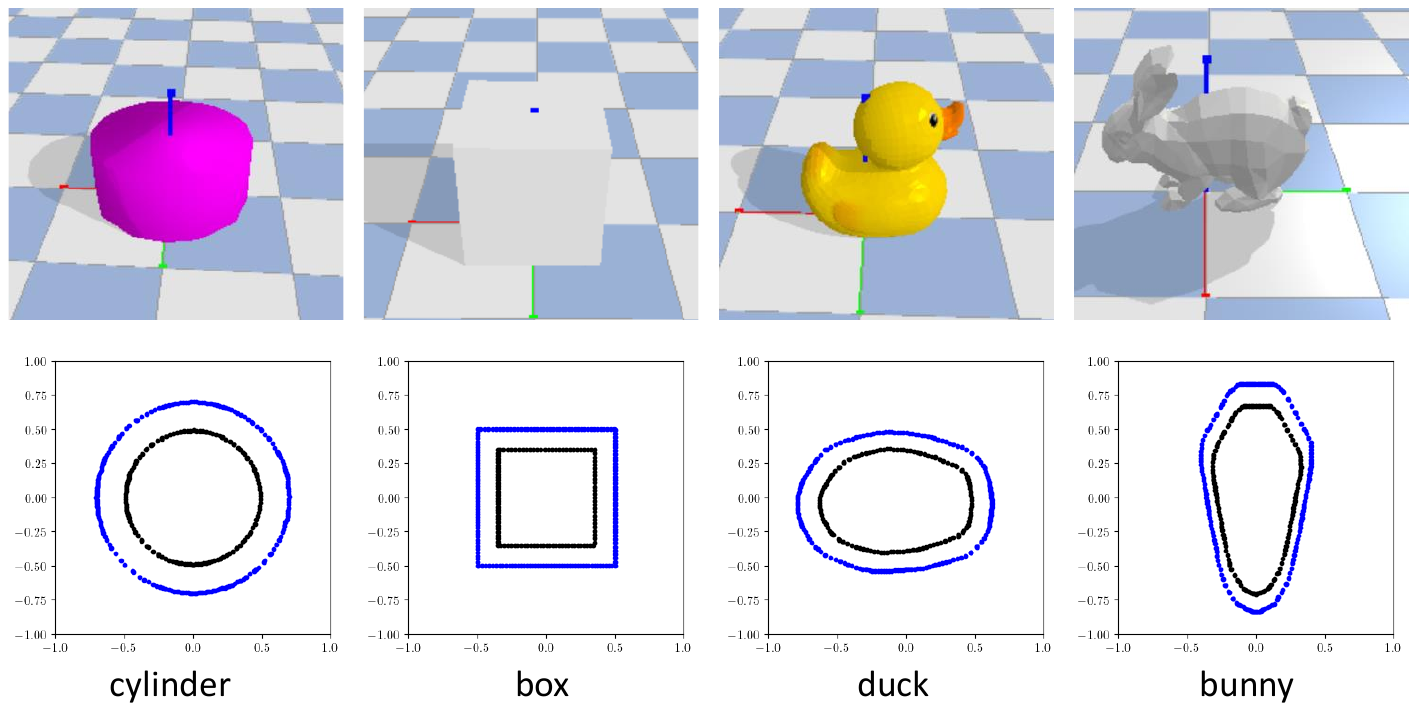}
 \end{center}
 \caption{PyBullet objects in case (II). The figures in the first and second lows are the GUI of the objects and corresponding cross sections at x-y coordinates, respectively. The surfaces with the smallest and largest sizes are shown by the black and blue lines, respectively.}
 \label{fig:shapes}
\end{figure}
\begin{figure*}[tb]
 \begin{center}
  \includegraphics[width=1\hsize]{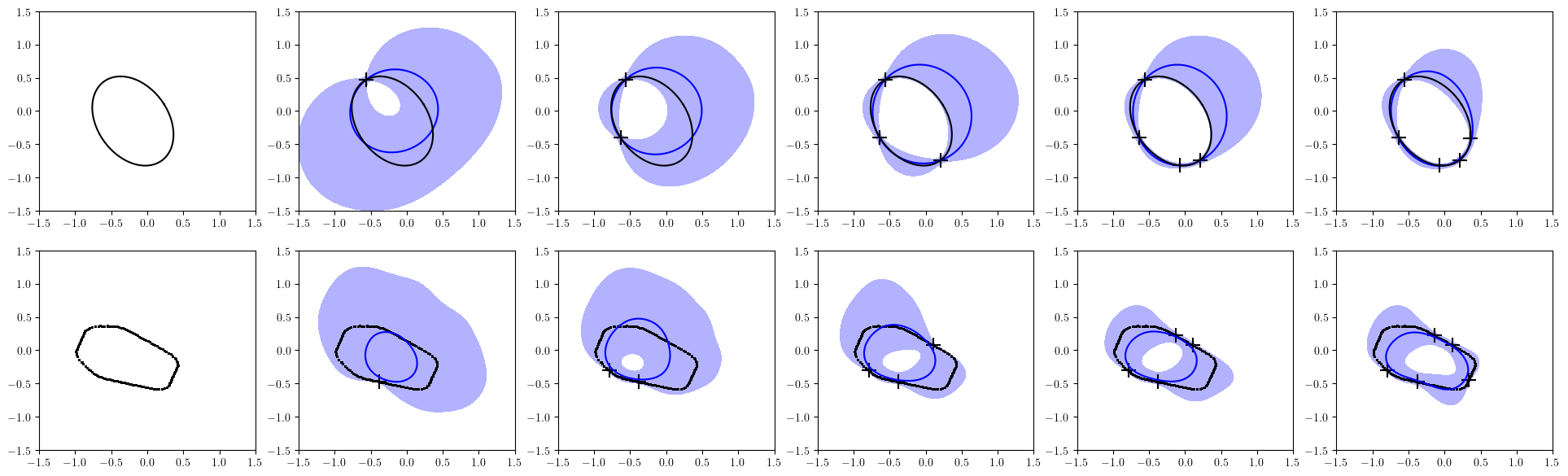}
 \end{center}
 \caption{A few examples of the prediction on unsafe regions: The black and blue lines show the actual surface and the mean prediction of the surface, respectively. 
 The blue shaded areas show the regions between the 0-level sets of the upper and lower 90\% bounds of $h_{\mathrm{IS}}$ derived by following the discussion in Section \ref{sec bound}, and the black cross marks show the data points on the surface used for the adaptation.}
 \label{fig:result1}
\end{figure*}
We first test the performance of the proposed scheme for constructing the IS function $h_{\mathrm{IS}}:\mathbb{R}^2\rightarrow \mathbb{R}$.
We consider 2 settings: (I) ellipsoidal obstacles with randomly chosen semi-axes lengths, center positions, and rotation angles (II) Objects in Fig. \ref{fig:shapes} with randomly chosen positions, rotation angles, and scales. 
In setting (I), the IS function $h_{\mathrm{IS}}$ is defined by the ellipsoids as follows:
\begin{align}
&h_{\mathrm{IS}}(z;\xi)=\frac{[(q_x-q_{x,0}^\xi)\cos\vartheta_{\mathrm{ob}}^\xi +(q_y-q_{y,0}^\xi)\sin\vartheta_{\mathrm{ob}}^\xi]^2}{c^2_{x,\xi}}\notag \\
&\quad +\frac{[(q_x-q_{x,0}^\xi)\sin\vartheta_{\mathrm{ob}}^\xi-(q_y-q_{y,0}^\xi)\cos\vartheta_{\mathrm{ob}}^\xi]^2}{c^2_{y,\xi}}- 1,
\end{align}
where $c_x^\xi$, $c_y^\xi\ (>0)$ are the semi-axes lengths along the $x$ and $y$ coordinates respectively, $\vartheta_{\mathrm{obs}}^\xi$ is the rotation angle, $(q_{x,0}^\xi,q_{y,0}^\xi)$ is the center of the ellipsoid, and $z=[q_x,q_y]^\top$. In this example, the unobserved parameter $\xi$ can be explicitly written as $\xi=[c_x^\xi,c_y^\xi,q_{x,0}^\xi,q_{y,0}^\xi,\vartheta_{\mathrm{obs}}^\xi]^\top$.
The parameters ($c_x^\xi$, $c_y^\xi$), ($q_{x,0}^\xi$, $q_{y,0}^\xi$), and  $\vartheta_{\mathrm{ob}}^\xi$ are randomly sampled from the uniform distribution in the range $[0.4,0.8]$, $[-0.8,0.8]$, and $[0,2\pi]$, respectively. In setting (II), we used object files in pybullet\_data \cite{pybulletData} to test with more arbitrary shapes. The rotational angles and the center position of the objects are sampled in the same way as setting (I). 
For both cases, the architecture of $\phi$ in (\ref{model}) is defined by three hidden layers of 256 units with 32 basis functions. Moreover, we set $N_{\mathrm{ite}}$, $\sigma$, $n_-$, $n_+$, and $\Delta$ to 30000, 0.001, 1, 5, and 0.1, respectively.

The prediction of an obstacle surface calculated with a small amount of data is shown in Figure \ref{fig:result1}.  
The blue lines and blue shaded areas in the figure are the mean predictions of the surfaces and the regions between the 0-level sets of the upper and lower 90\% bounds of $h_{\mathrm{IS}}$. We can see that the confidence bounds are reasonably tight even when the number of data points for the adaptation is small, which is the main effect of the offline meta-training. Moreover, same as reported in \cite{ALPaCA}, the time required for the adaptation is much smaller than that of the GP \cite{GP1} which is a commonly used Bayesian method.  
More results including the comparison to the GP in terms of both the accuracy of the prediction and time required for the training are provided in Appendix \ref{appendix:prediction}.

\subsection{Control execution}
\begin{figure}[htbp]
  \centering
  \subfloat{\includegraphics[width=0.475\textwidth]{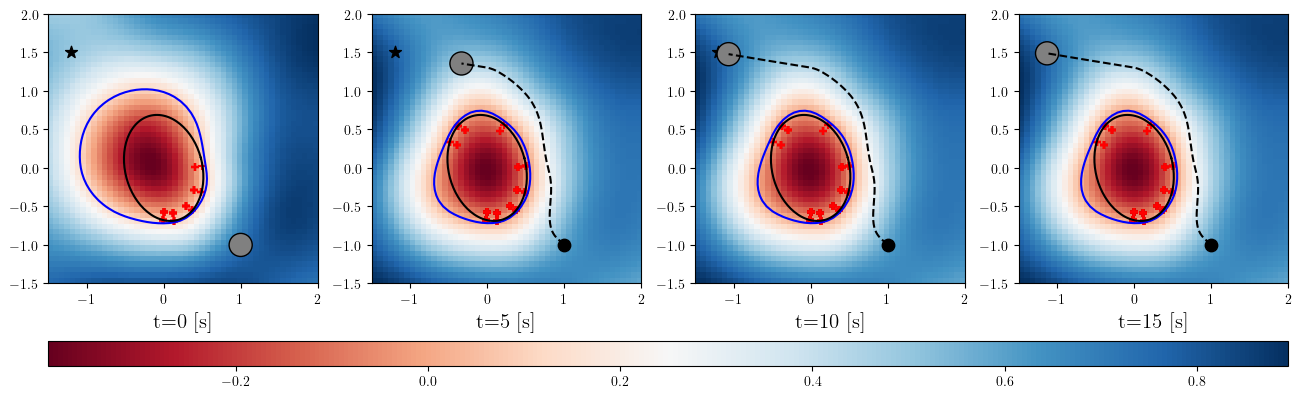}\label{fig:subfig1}}
  \hfill
  \subfloat{\includegraphics[width=0.475\textwidth]{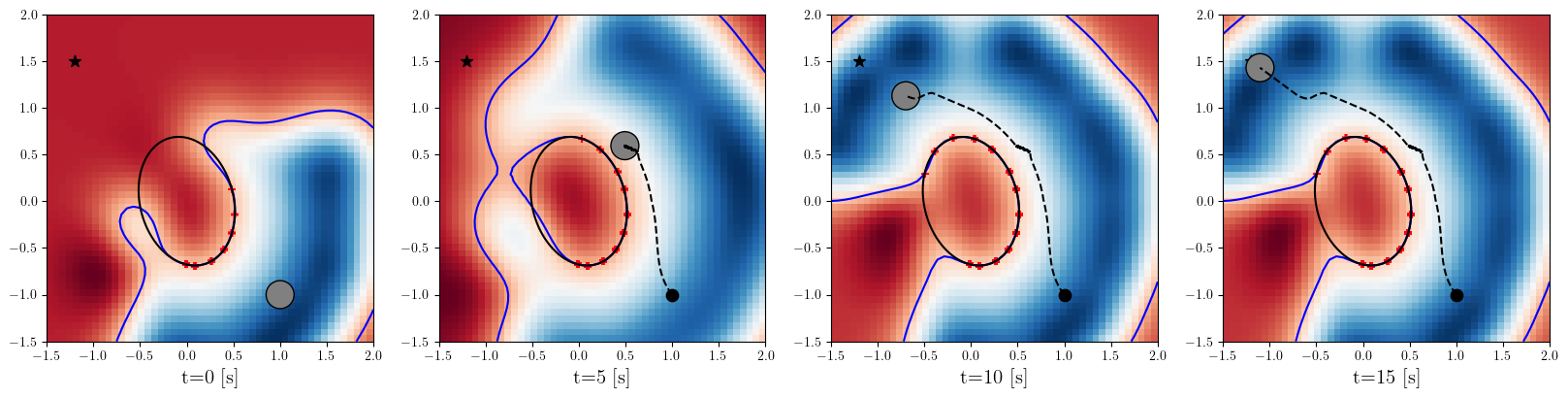}\label{fig:subfig2}}
  \caption{An example of control execution with the proposed method (above) and GP case (below). The red crosses, heat maps, blue solid lines, black solid lines, and black dotted lines represent data points on the surface, the values of CBF, 0-level sets of 95\% lower bounds of CBF, actual surfaces, and robot trajectories, respectively. The CBF is updated every 5 [s].}
  \label{fig:execution}
\end{figure}
Based on the trained $h_{\mathrm{IS}}$, we next consider executing the control with Algorithm \ref{alg}. 
The control objective is to steer the vehicle toward the goal position $(q_x^g,q_y^g)$ while avoiding obstacles. The goal-reaching objective is encoded through the CLF $V(x)= (q_x-q_x^g)^2+(q_y-q_y^g)^2$.
Moreover, $\alpha_C$ and $\alpha_V$ in (\ref{qp}) are defined by the linear functions $\alpha_C(c)=\gamma_Cc$ and $\alpha_V(c)=\gamma_Vc$. We set parameters $\gamma_C$, $\gamma_V$, and $\lambda$ to $1.0$, $1.0$, and 10, respectively. 
For comparison, we test the control performance of the proposed scheme and the case where the GP is used instead of ALPaCA. 
Since the theoretically guaranteed bounds of GP \cite{GP2,GP3} tend to be overly conservative, we instead use the 2-$\sigma$ lower bounds in the GP case.
The metric of the control performance is defined by the cumulative squared error (CSE) of the difference between the robot and goal positions, along with the resulting state trajectory (the squared errors are corrected every 0.02 [s]). 
An example of the control execution with $\Delta_{\mathrm{lidar}}=5$ [s] is visualized in Figure \ref{fig:execution}. 
From Figure \ref{fig:execution}, we can see the vehicle successfully reaches the goal while avoiding the unsafe region. Notably, the control of the proposed method is much less conservative than the GP case because of the tight prediction of the unsafe region (the vehicle in the GP case is often forced to be stacked in small safe regions until the subsequent update of the CBF). The CSEs in this case for the proposed method and GP case are 618 and 1010, respectively. We additionally provide results for different 5 environments with single or multiple obstacles in Appendix \ref{appendix:control} (the results for the cases $\Delta_{\mathrm{lidar}}=1,3,5$ are shown) and similar results are observed. 
Although the CSE of the proposed method and the GP case get closer as $\Delta_{\mathrm{lidar}}$ becomes smaller, it is ideal to maintain the value $\Delta_{\mathrm{lidar}}$ to be small to save computation and the use of sensor devices that have limited batteries. Moreover, the GP prediction can not be updated so frequently because of the time required for the training (see the time analysis in Appendix \ref{appendix:prediction}).
From the results above, we can conclude that the proposed scheme has a practical advantage from the perspective of control. 

\section{Conclusion}
In this paper, we proposed an efficient and probabilistically ensured online learning method of CBF with sensor measurements. 
In the proposed method, we utilized a technique based on the GPIS to train unsafe regions from the measurements. Specifically, a Bayesian meta-learning scheme was employed to learn IS function, which enables us to effectively use past data from different settings. 
To avoid the volume of data buffers becoming unnecessarily large in the online control execution, we also considered the data selection scheme which effectively uses the uncertainty information provided by the current model. The prediction made by the proposed scheme was readily incorporated in the CBF-CLF-QP formulation and the probabilistic safety guarantee for the control scheme was also provided. In the case study, we have shown the efficacy of our method in terms of the conservativeness of the prediction, the time required for the online training, and the conservativeness of the control performance. 

\onecolumn
\appendix
\subsection{Pseudo code of offline meta-learning}\label{pseudo}
The Pseudo code of offline meta-learning is summarized in Algorithm \ref{alg2}.
\begin{algorithm}[t]\label{alg2}
{\small
\SetKwInOut{Input}{Input}
\SetKwInOut{Output}{Output}
\Input{$p(\xi)$ (distribution of obstacles);$\sigma$ (noise variance);}
Randomly initialize the parameters $\bar{\theta}_0,\Lambda_0,w$\\

\For{$i=1:N_{\mathrm{ite}}$}{
\For{$j=1:J$}{
Sample obstacles $\xi_j \sim p(\xi)$;\\
Construct datasets corresponding to the sampled obstacles as $D_{\xi_j}=\cup_{i=1}^{n_j}\{(z_i^j, y_i^j)\}$;\\
Sample $n_j^{\mathrm{tr}}$ from uniform distribution over $\{ 1,\ldots,n_j \}$\\
Split the dataset into training and test sets: $D_{\xi_j}^{\mathrm{tr}}=\cup_{i=1}^{n_{j}^{\mathrm{tr}}}\{(z_i^{\mathrm{tr},j}, y_i^{\mathrm{tr},j})\}$ and $D_{\xi_j}^{\mathrm{ts}}=\cup_{i=1}^{n_{j}^{\mathrm{ts}}}\{(z_i^{\mathrm{ts},j}, y_p^{\mathrm{ts},j})\}$\\
Compute posterior parameters $\bar{\theta}_{\xi_j}$ and $\Lambda_{\xi_j}$ via (\ref{post}) with data $D_{\xi_j}^{\mathrm{tr}}$\\
}
Update $\bar{\theta}_0,\Lambda_0,w$ via gradient step on the loss (\ref{loss}) calculated by $\bar{\theta}_{\xi_j}$, $\Lambda_{\xi_j}$, and $D_{\xi_j}^{\mathrm{ts}}$, $\forall j \in \mathbb{N}_{1:J}$} 

    \caption{Offline meta-learning} 
    }
\end{algorithm}

\subsection{Mitigation of conservativeness}
To avoid unnecessary large $\beta_\xi$ defined in Section \ref{sec bound} and mitigate the conservativeness of the probabilistic bounds without compromising on safety guarantees, we can add the following regularization term in the original loss function (\ref{loss}) as proposed in \cite{ALPaCA2}. 
\begin{align}\label{reg}
\mathcal{L}_{\mathrm{reg}}^\xi(\Lambda_0) = \gamma \mathrm{Tr}(\Lambda_\xi^{-T} \Lambda_\xi^{-1})\mathrm{Tr}(\Lambda_0^{-T} \Lambda_0^{-1}),
\end{align}
where $\gamma\in \mathbb{R}_{\geq 0}$ represents the weights for the regularization term and $\mathrm{Tr}(\cdot)$ is the determinant of the argument matrix. This term corresponds to the upper bound on the term $\frac{\bar{\lambda}(\Lambda_0)}{\underline{\lambda}(\Lambda_\xi)}$ in $\beta_\xi$. Thus, minimizing (\ref{reg}) leads to small $\beta_\xi$ value.
In the experiment, we have confirmed that the regularization term with $\gamma=10^{-9}$ reasonably reduces the conservativeness of the bounds.
\subsection{Detailed Results for Meta-Learning of Unsafe Regions}\label{appendix:prediction}
We first show the negative log-likelihood (NLL) of the predictions for the cases (I) and (II) discussed in Section \ref{case study}, in Figure \ref{fig:nll}. The NLLs against the number of data points are calculated for both the proposed scheme and the case where we use Gaussian process regression (GPR) instead of ALPaCA. NLLs are calculated for 100 different obstacles and the average and 3-$\sigma$ intervals are plotted in the figures. Note that the number of data points in the figures means the number of surface points and the actual number of data points is $n_++n_-+1$ ($=7$) times of that. The squared exponential kernel is used for GP and the parameters in the kernel are fitted by solving the maximum likelihood problem. In GP implementation, the Python library GPy \cite{GPy} is used. From the figures, we can see that the prediction made by the proposed meta-learning scheme is superior to that of the GPR case when the number of data is small. The time analysis is also provided in Figure \ref{fig:time}. Figure \ref{fig:time} shows that the proposed scheme can produce the prediction very fast compared to the case in which every time fits the GP, which is especially beneficial for the usage of the online synthesis of CBF-QP since we can more frequently update the CBFs. Additional examples of the prediction are shown in Figure \ref{fig:resultAdd1} and \ref{fig:resultAdd2}.

\begin{figure*}[htbp]
  \centering
  \subfloat[Setting (I)]{\includegraphics[width=0.45\textwidth]{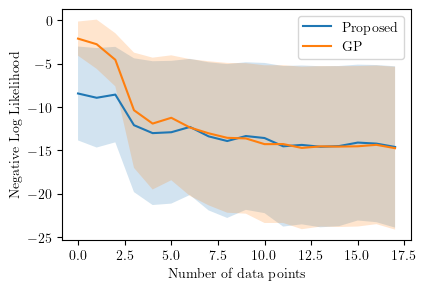}}
  \hfill
  \subfloat[Setting (II)]{\includegraphics[width=0.45\textwidth]{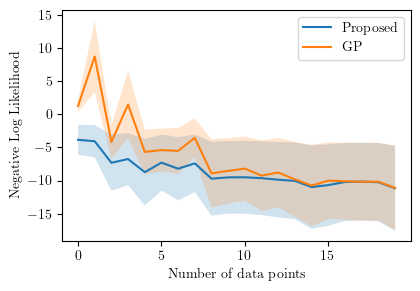}}
  \caption{NLLs for predictions of the proposed meta-learning scheme and GP case against the number of data points (Note that the number of data points in the figure means the number of surface points and the actual number of data points is $n_++n_-+1$ ($=7$) times of that.). The NLLs are collected for randomly chosen 100 different functions and the mean and $3-\sigma$ interval of NLL values are plotted in the figure.}
  \label{fig:nll}
\end{figure*}

\begin{figure*}[thbp]
\begin{center}

\includegraphics[width=0.6\hsize]
{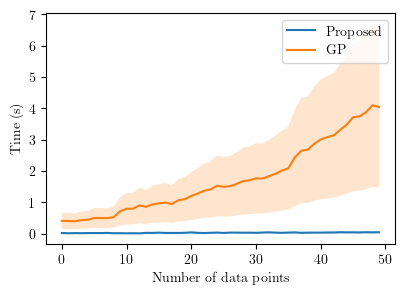}

\end{center}
\caption{The time required for predictions of the proposed meta-learning scheme and GP case against the number of data points (Note that the number of data points in the figure means the number of surface points and the actual number of data points is $n_++n_-+1$ ($=7$) times of that.). The times are collected for randomly chosen 100 different functions and the mean and $3-\sigma$ interval of NLL values are plotted in the figure.}
\label{fig:time}
\end{figure*}

\begin{figure*}[htbp]
  \centering
  \subfloat[Additional visualization for case (I)]{\includegraphics[width=1\textwidth]{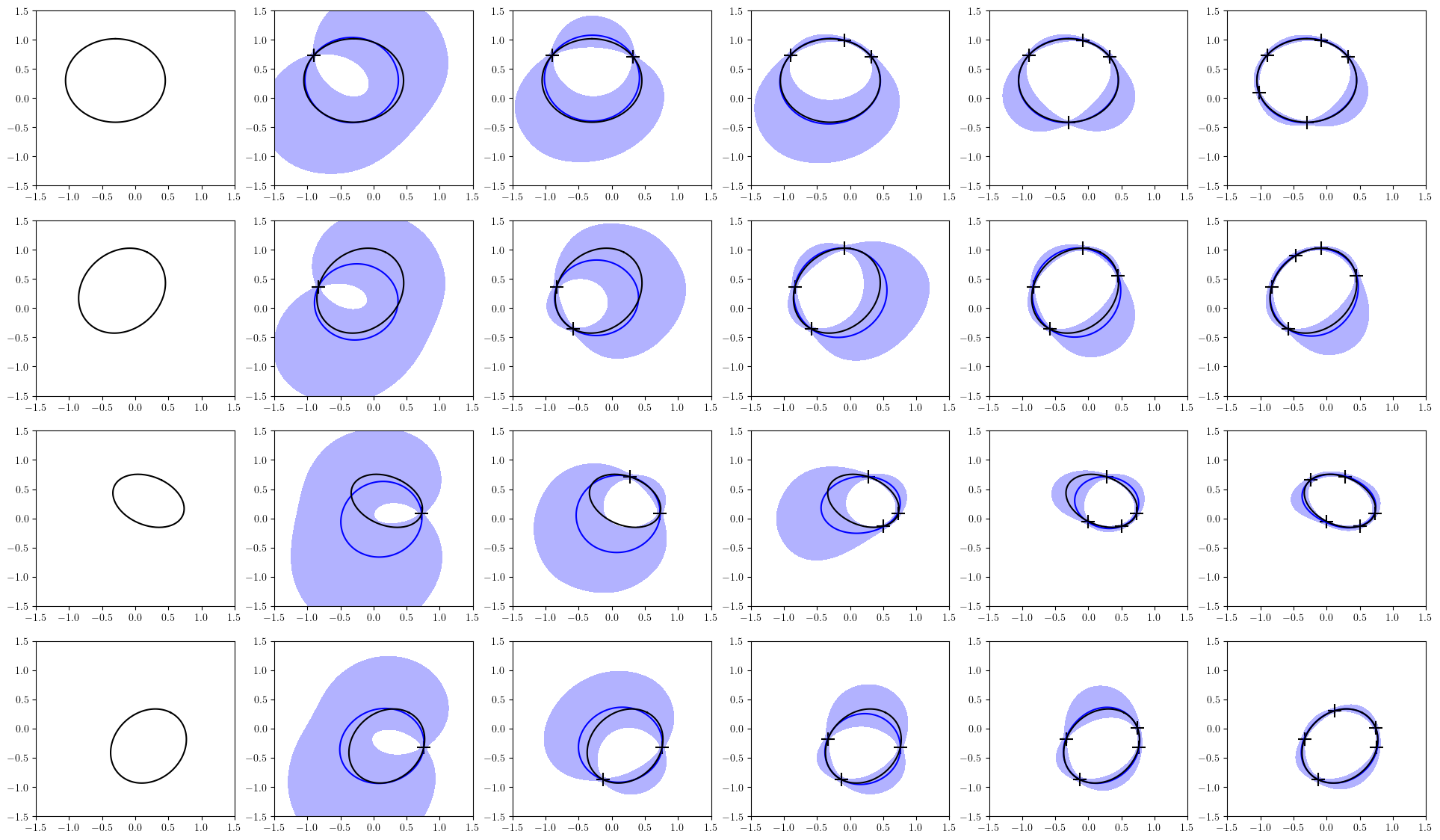}\label{fig:resultAdd1}}
  \hfill
  \subfloat[Additional visualization for case (II)]{\includegraphics[width=1\textwidth]{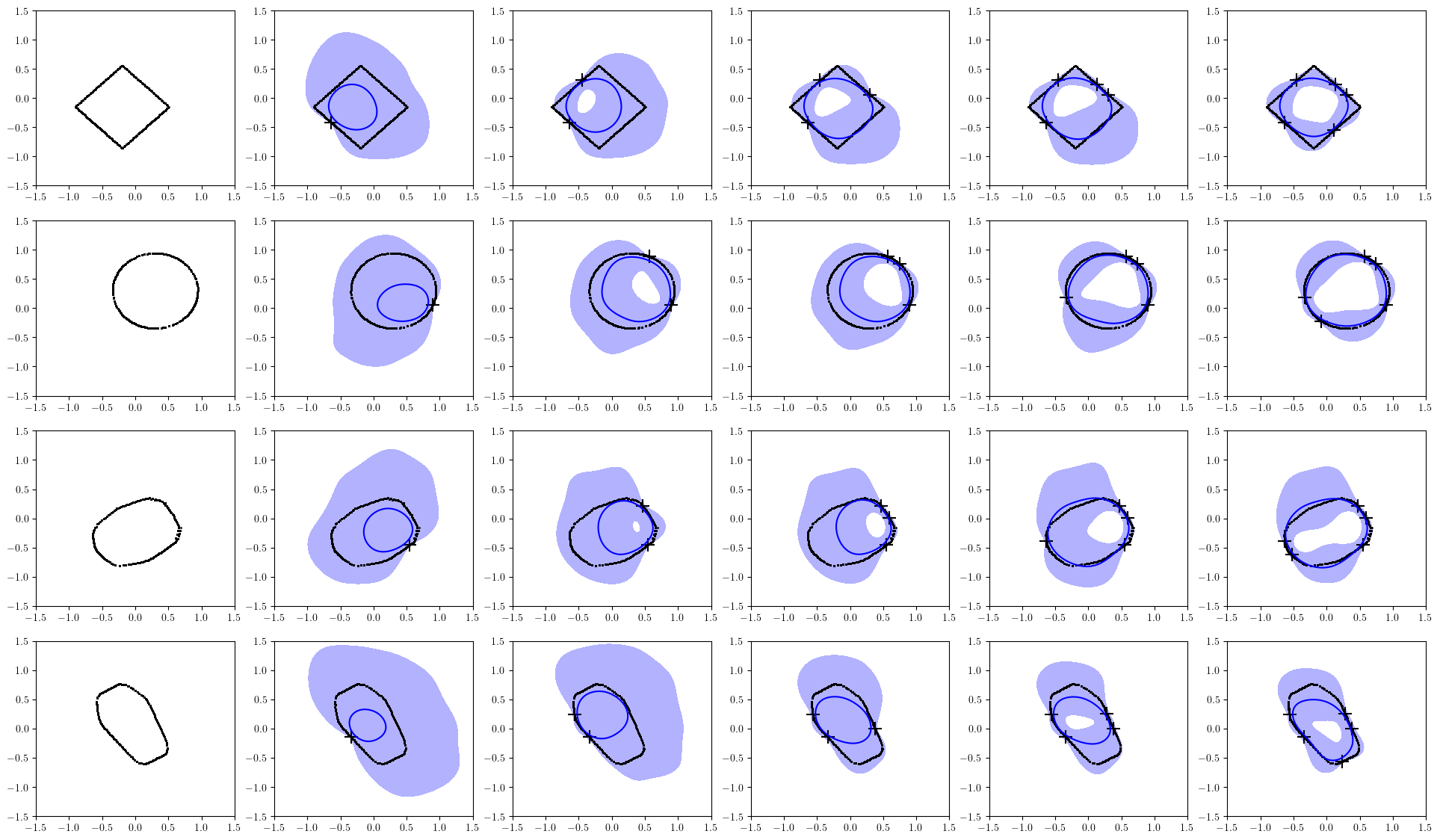}\label{fig:resultAdd2}}
\end{figure*}

\subsection{Detailed Results for Control Execution}\label{appendix:control}
Next, we summarize the detailed results for control execution. In Figures \ref{fig:execution add1}-- \ref{fig:execution add5}, we show the resulting trajectories and the time evolution of the squared errors between the robot positions and the goal positions when $\Delta_{\mathrm{lidar}}=5$, for all 6 environments.
From these figures, we can see that the tight prediction of the proposed scheme leads to less conservative control performance as explained in Section \ref{case study}.
In Table \ref{CSE}, we also summarize the cumulative squared errors between the robot position and goal position for all environments and the cases $\Delta_{\mathrm{lidar}}=1,3,5$ [s].
\begin{figure}[htbp]
  \centering
  \subfloat{\includegraphics[width=0.95\textwidth]{DataForRALpropose.png}}
  \hfill
  \subfloat{\includegraphics[width=0.95\textwidth]{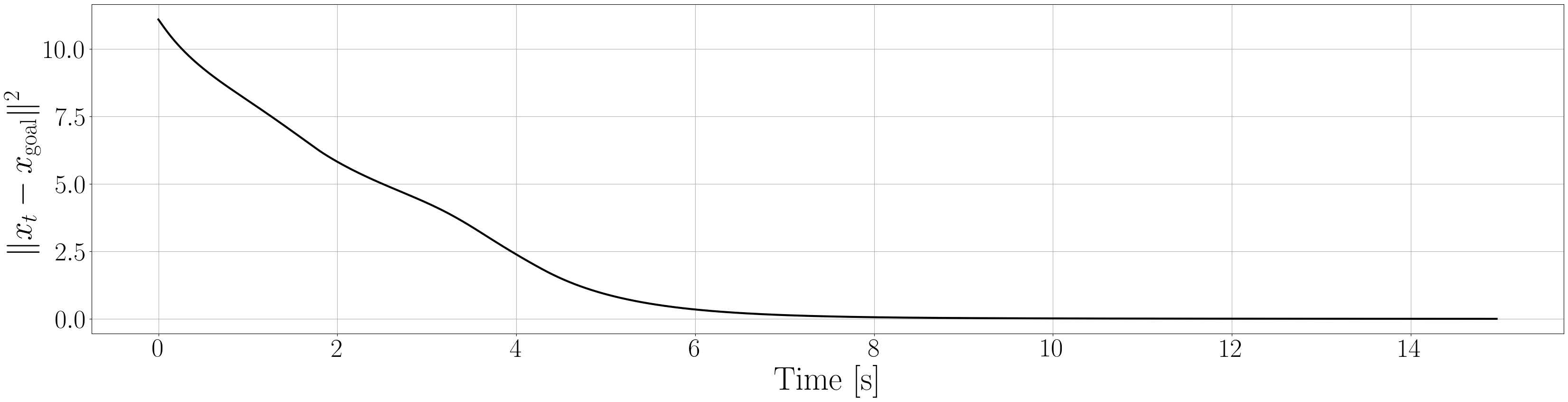}}
  \hfill
  \subfloat{\includegraphics[width=0.95\textwidth]{DataForRALgp.png}}
  \hfill
  \subfloat{\includegraphics[width=0.95\textwidth]{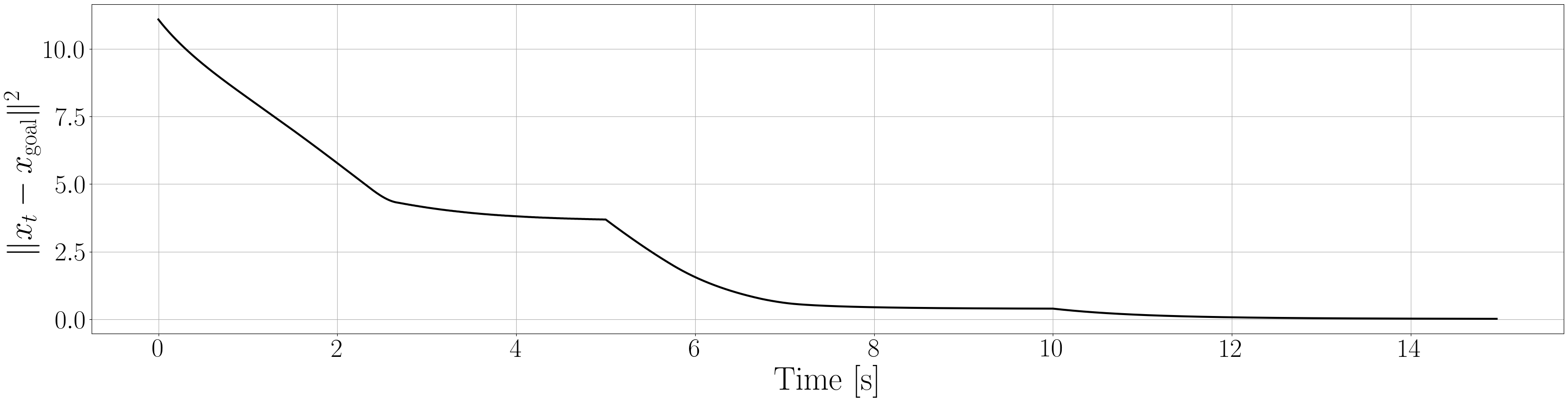}}
  \caption{A control execution in Environment 1 with the proposed method (above) and GP case (below). The red crosses, heat maps, blue solid lines, black solid lines, and black dotted lines represent data points on the surface, the values of CBF, 0-level sets of 95\% lower bounds of CBF, actual surfaces, and robot trajectories, respectively. The prediction of CBF is updated every 5 [s]. The time evolution of the squared errors between goal and system states is also shown in the figure.}
  \label{fig:execution add1}
\end{figure}

\begin{figure}[htbp]
  \centering
  \subfloat{\includegraphics[width=0.95\textwidth]{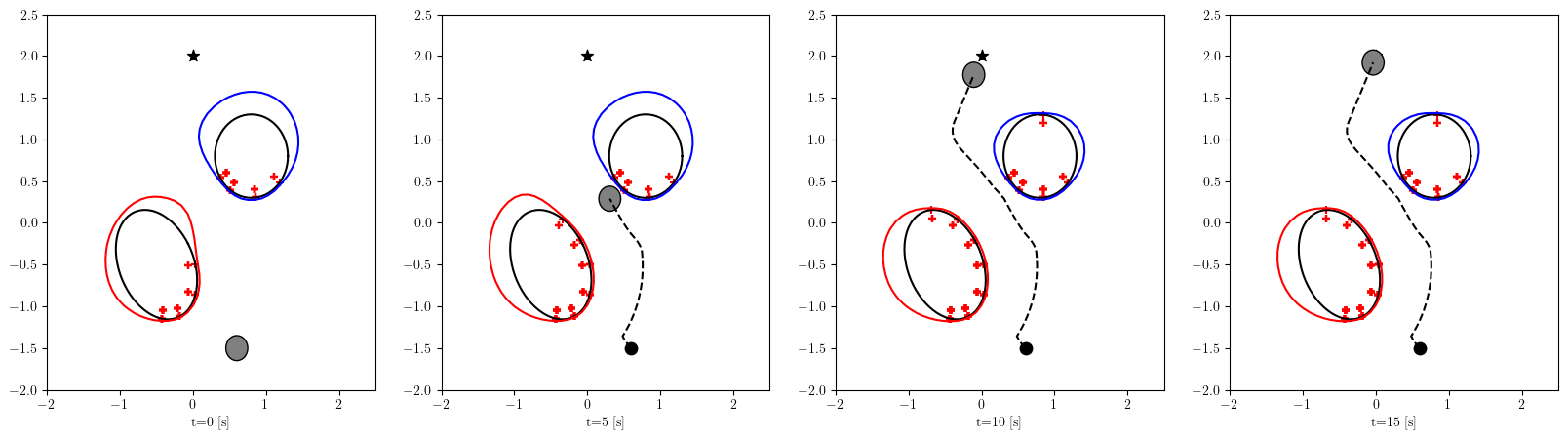}}
  \hfill
  \subfloat{\includegraphics[width=0.95\textwidth]{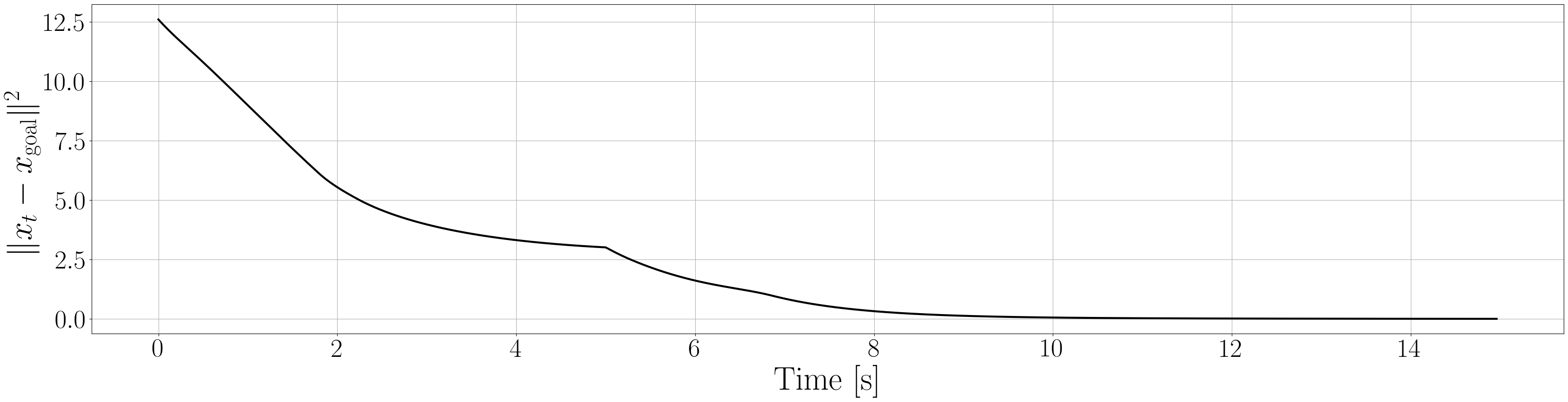}}
  \hfill
  \subfloat{\includegraphics[width=0.95\textwidth]{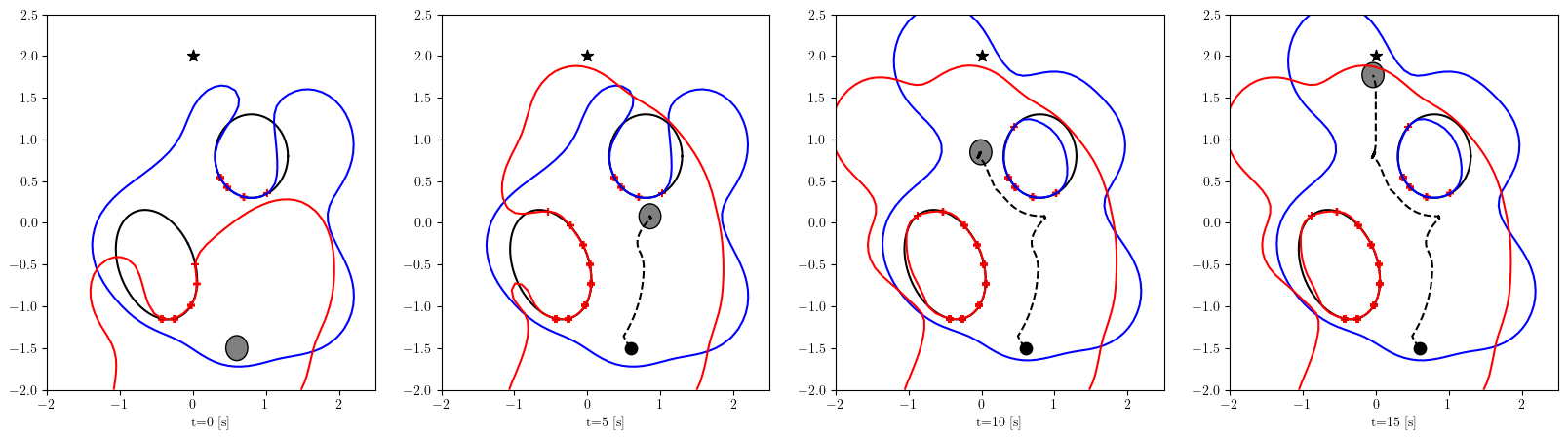}}
  \hfill
  \subfloat{\includegraphics[width=0.95\textwidth]{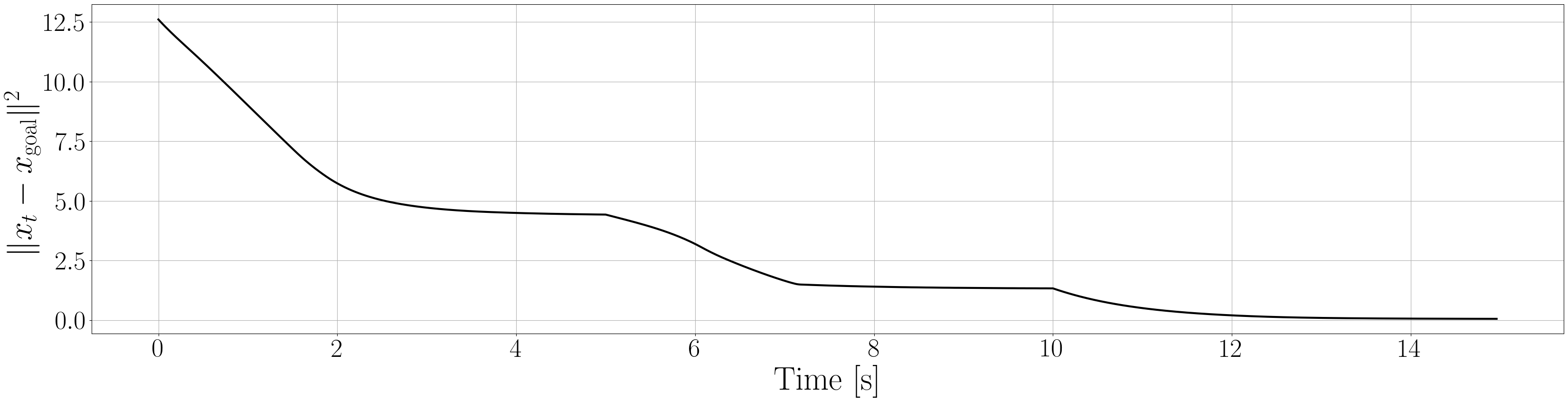}}
  \caption{A control execution in Environment 2 with the proposed method (above) and GP case (below). The red crosses, blue/red solid lines, black solid lines, and black dotted lines represent data points on the surface,  0-level sets of 95\% lower bounds of CBF, actual surfaces, and robot trajectories, respectively. The prediction of CBF is updated every 5 [s]. The time evolution of the squared errors between goal and system states is also shown in the figure.}
  \label{fig:execution add2}
\end{figure}

\begin{figure}[htbp]
  \centering
  \subfloat{\includegraphics[width=0.95\textwidth]{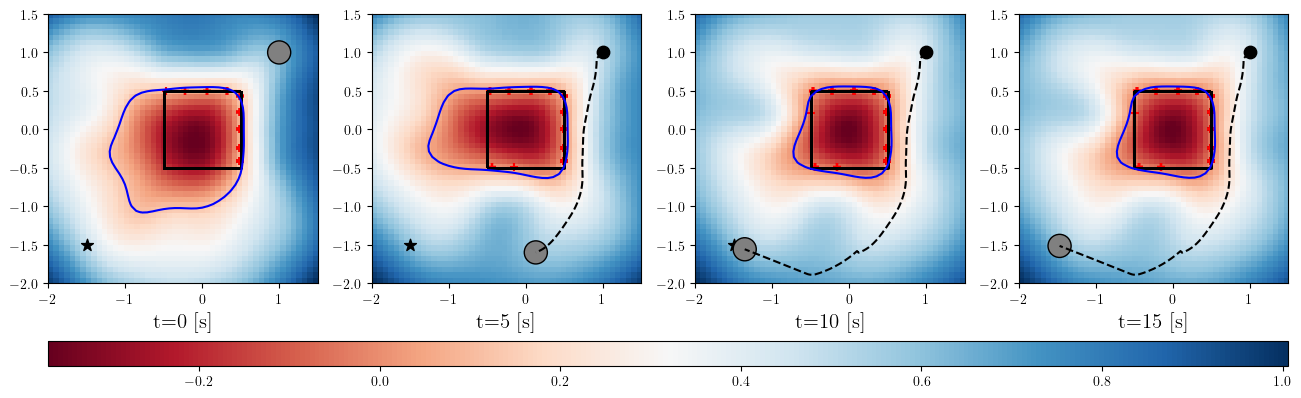}}
  \hfill
  \subfloat{\includegraphics[width=0.95\textwidth]{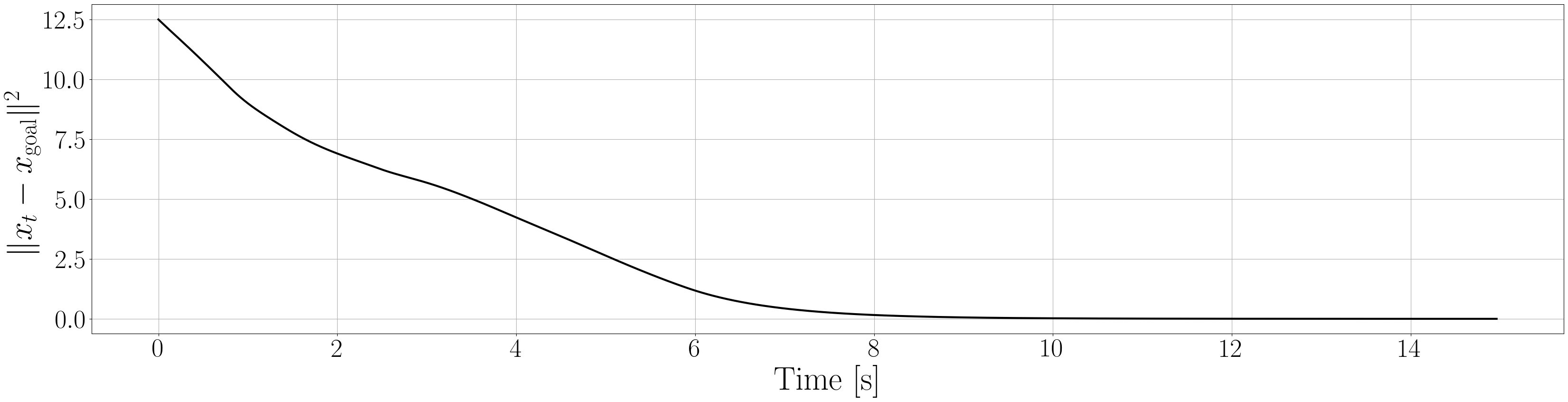}}
  \hfill
  \subfloat{\includegraphics[width=0.95\textwidth]{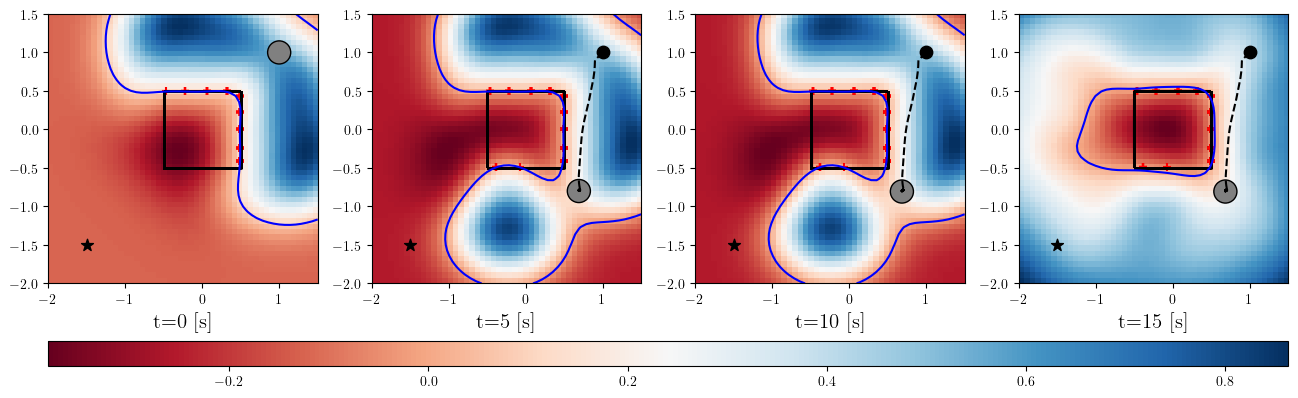}}
  \hfill
  \subfloat{\includegraphics[width=0.95\textwidth]{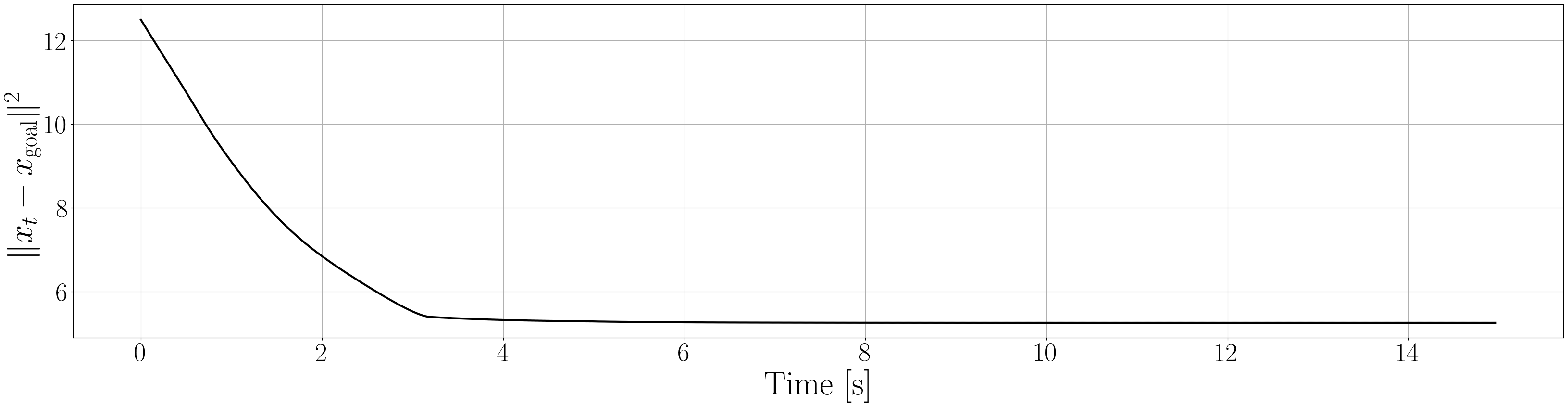}}
  \caption{A control execution in Environment 3. The prediction of CBF is updated every 5 [s].}
  \label{fig:execution add3}
\end{figure}

\begin{figure}[htbp]
  \centering
  \subfloat{\includegraphics[width=0.95\textwidth]{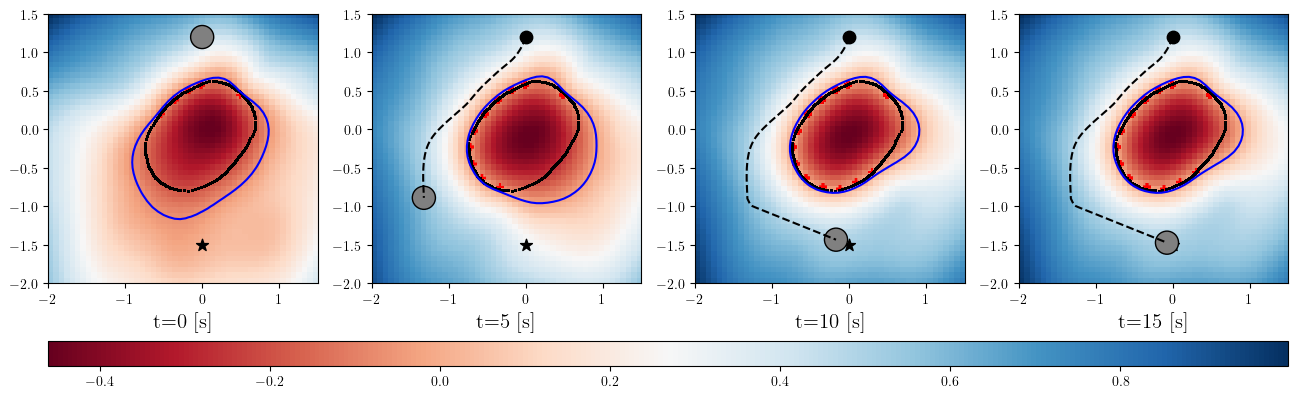}}
  \hfill
  \subfloat{\includegraphics[width=0.95\textwidth]{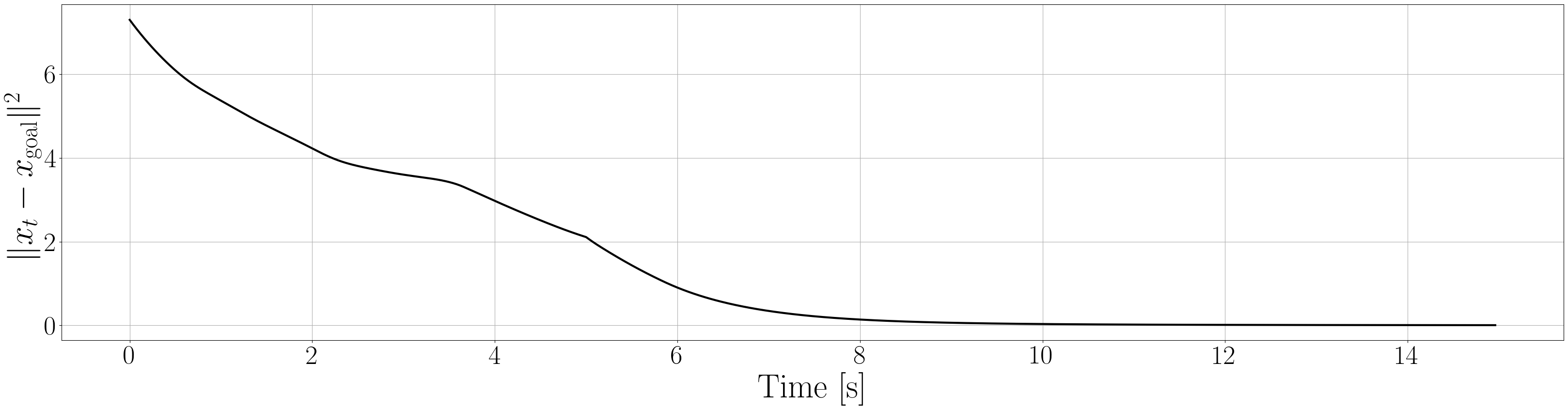}}
  \hfill
  \subfloat{\includegraphics[width=0.95\textwidth]{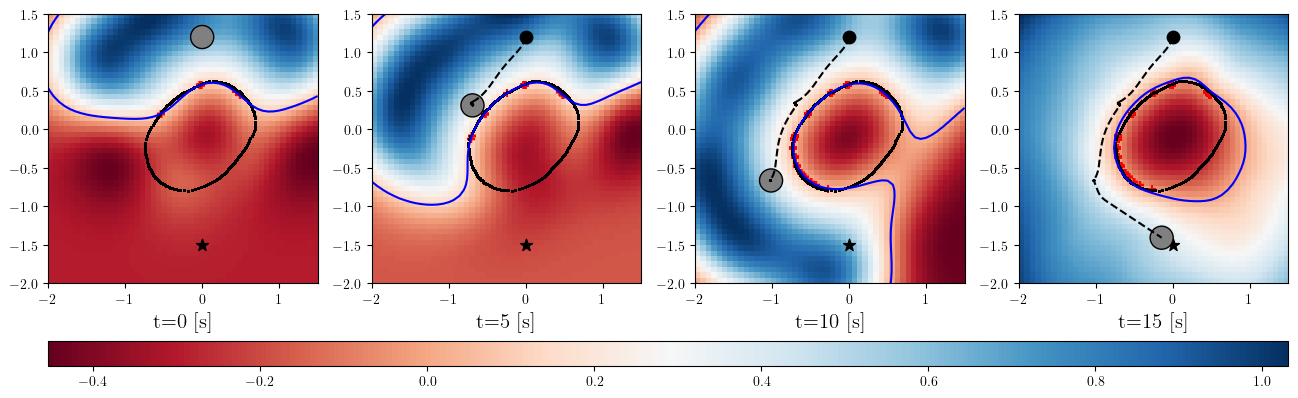}}
  \hfill
  \subfloat{\includegraphics[width=0.95\textwidth]{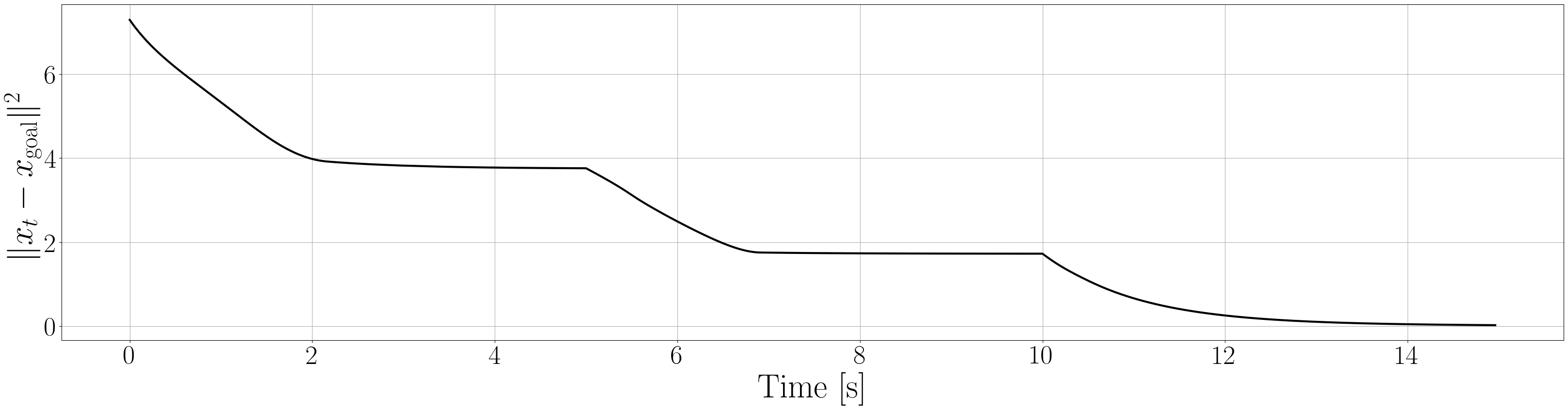}}
  \caption{A control execution in Environment 4. The prediction of CBF is updated every 5 [s].}
  \label{fig:execution add4}
\end{figure}

\begin{figure}[htbp]
  \centering
  \subfloat{\includegraphics[width=0.95\textwidth]{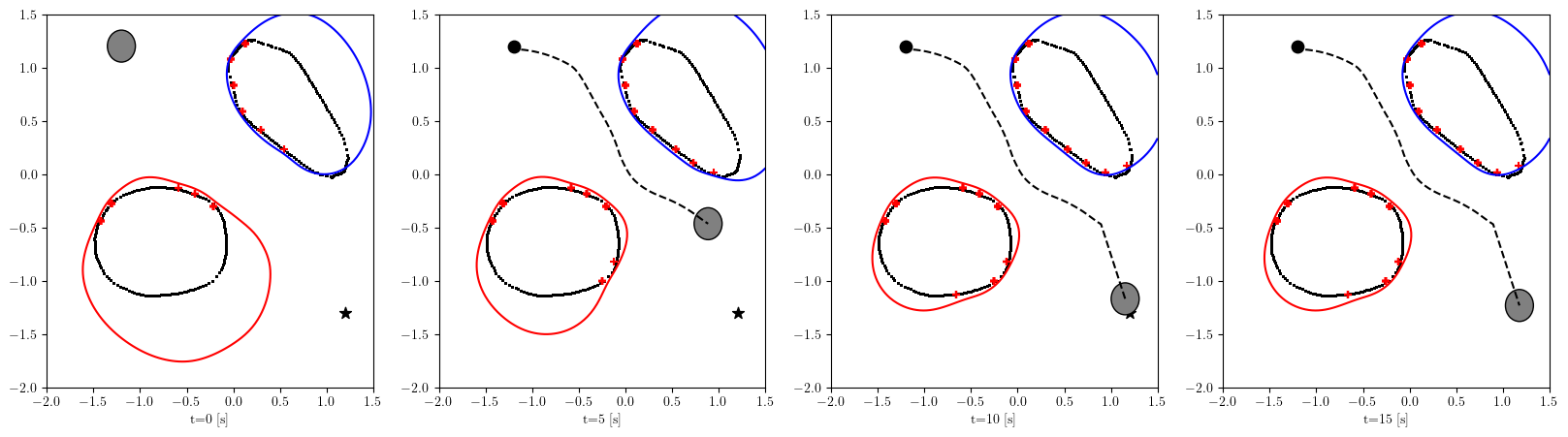}}
  \hfill
  \subfloat{\includegraphics[width=0.95\textwidth]{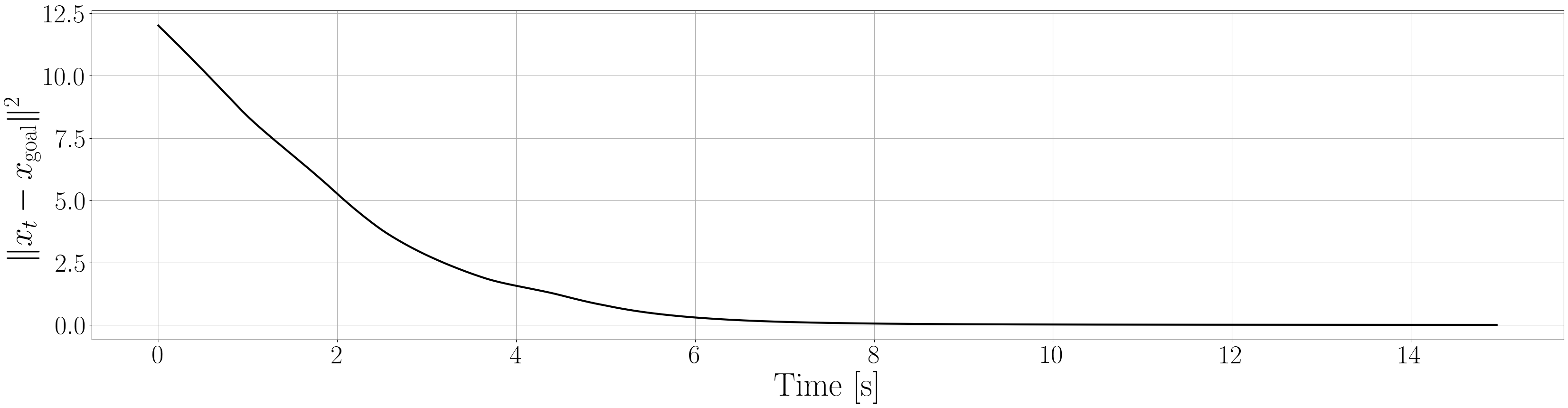}}
  \hfill
  \subfloat{\includegraphics[width=0.95\textwidth]{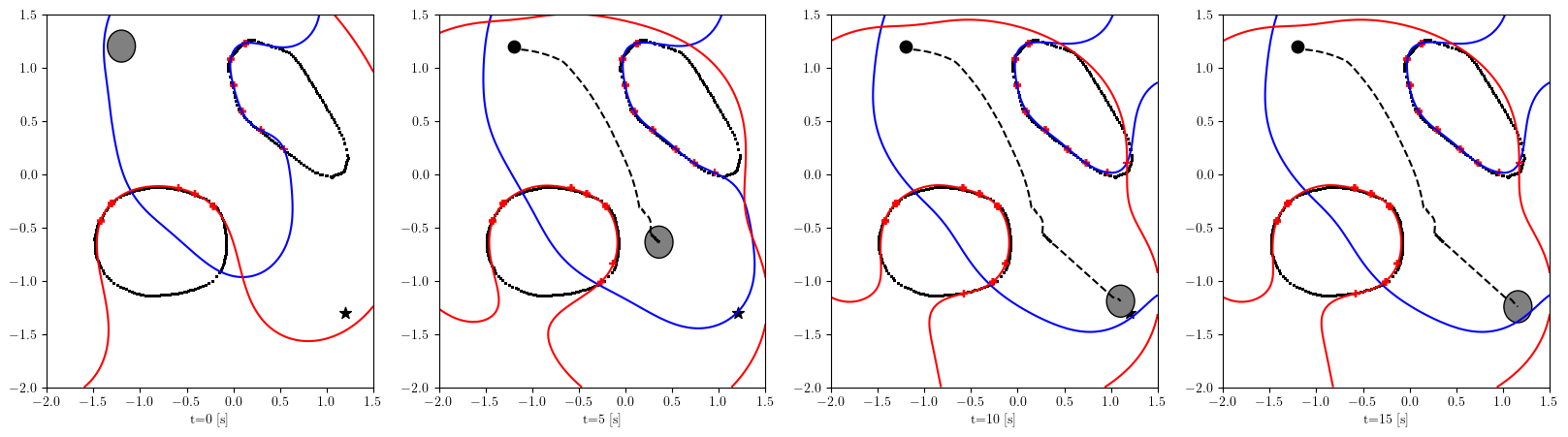}}
  \hfill
  \subfloat{\includegraphics[width=0.95\textwidth]{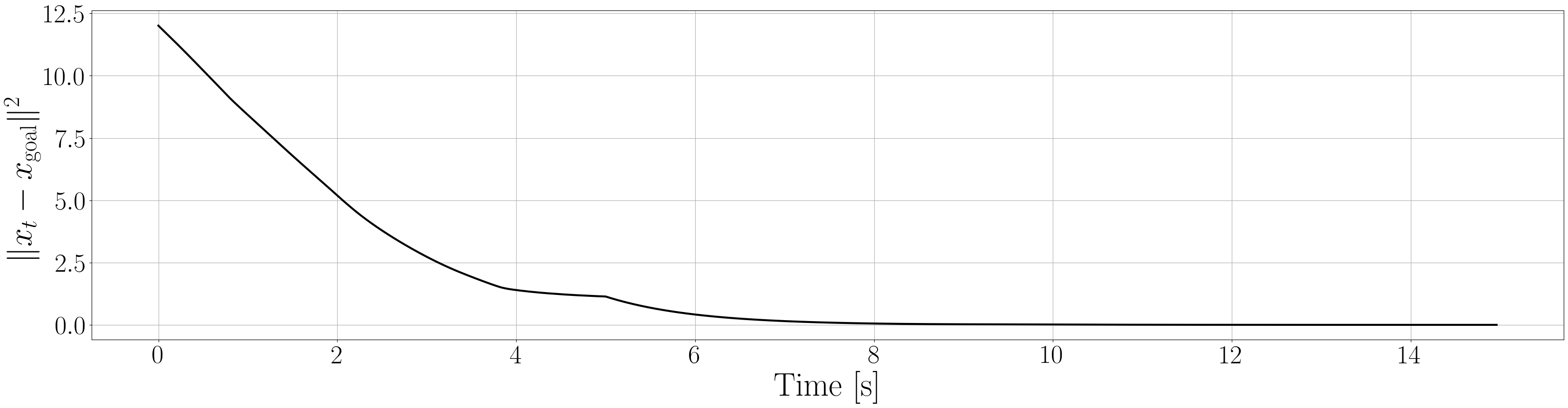}}
  \caption{A control execution in Environment 5. The prediction of CBF is updated every 5 [s].}
  \label{fig:execution add5}
\end{figure}

\begin{figure}[htbp]
  \centering
  \subfloat{\includegraphics[width=0.95\textwidth]{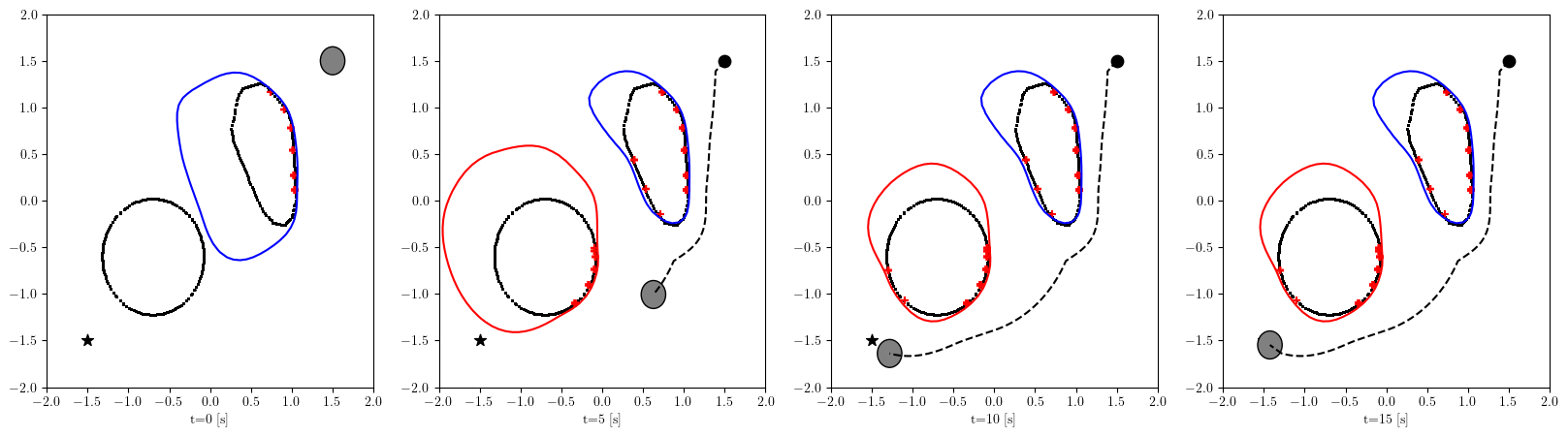}}
  \hfill
  \subfloat{\includegraphics[width=0.95\textwidth]{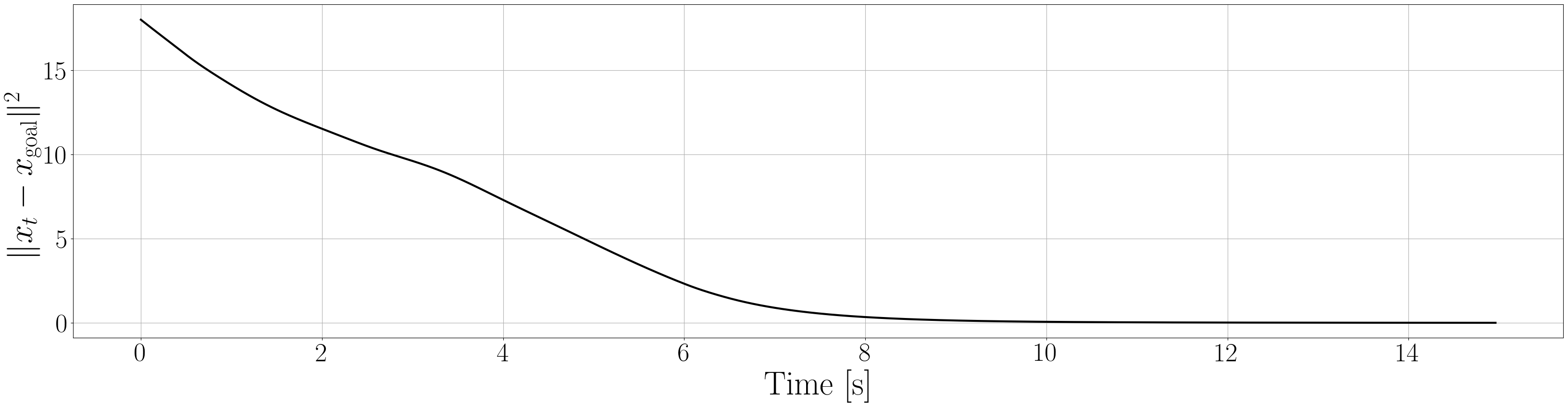}}
  \hfill
  \subfloat{\includegraphics[width=0.95\textwidth]{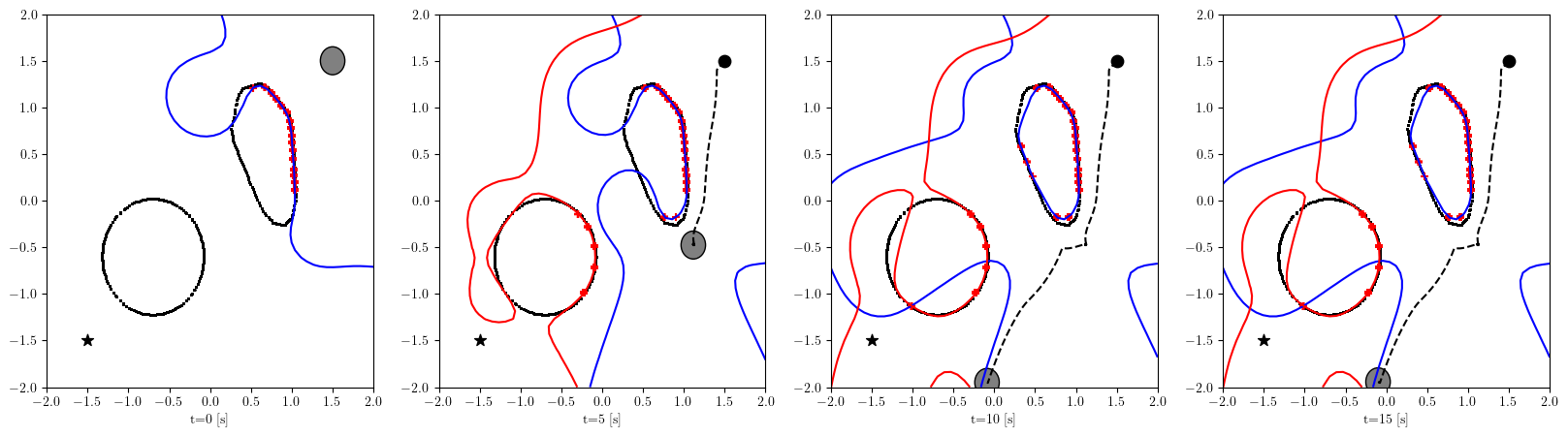}}
  \hfill
  \subfloat{\includegraphics[width=0.95\textwidth]{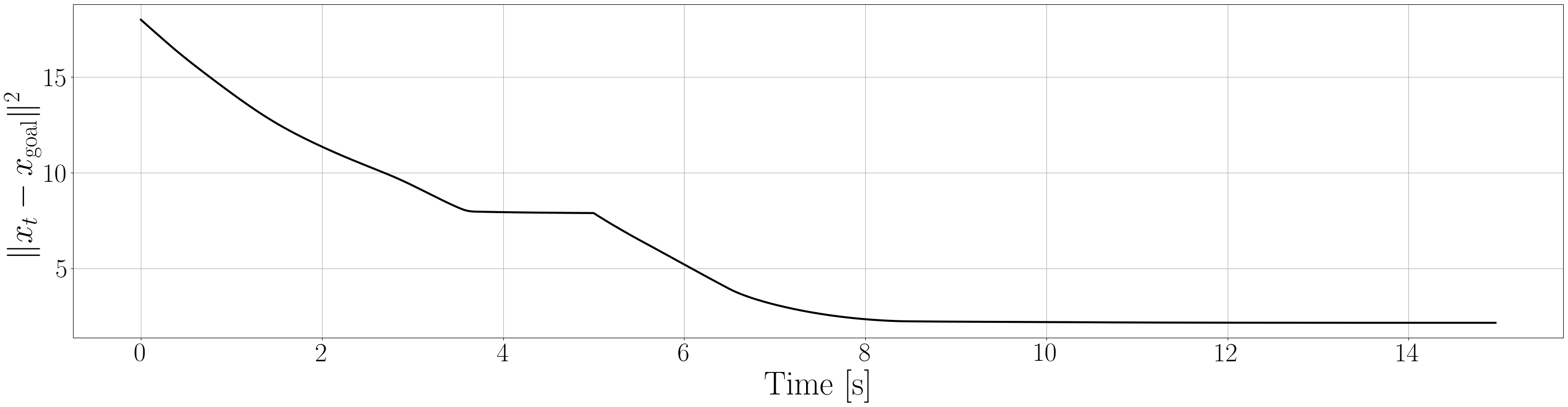}}
  \caption{A control execution in Environment 6. The prediction of CBF is updated every 5 [s].}
  \label{fig:execution add6}
\end{figure}

\newpage
\begin{table}[]

\centering
\caption{The cumulative squared errors between the robot position and the goal position}
\begin{tabular}{c lll lll}\label{CSE}
\\ \hline
\multicolumn{1}{|l|}{} & \multicolumn{3}{|c}{Proposed} & \multicolumn{3}{|c|}{GP}                                       \\ \hline
\multicolumn{1}{|l|}{Environment / $\Delta_{\mathrm{lidar}}$} & \multicolumn{1}{c|}{1}     & \multicolumn{1}{c|}{3}     & \multicolumn{1}{|c}{5} & \multicolumn{1}{|c|}{1}     & \multicolumn{1}{c|}{3}     & \multicolumn{1}{c|}{5} \\ \hline
\multicolumn{1}{|c|}{1}                                   & \multicolumn{1}{l|}{554.93} & \multicolumn{1}{l|}{620.12} & 618.56                  & \multicolumn{1}{|l|}{543.78} & \multicolumn{1}{l|}{715.75} & \multicolumn{1}{l|}{1010.93}               \\ \hline
\multicolumn{1}{|c|}{2}                                    & \multicolumn{1}{l|}{654.10} & \multicolumn{1}{l|}{702.75} & 830.22                  & \multicolumn{1}{|l|}{697.65} & \multicolumn{1}{l|}{818.47} & \multicolumn{1}{l|}{1094.79}               \\ \hline
\multicolumn{1}{|c|}{3}                                    & \multicolumn{1}{l|}{534.40} & \multicolumn{1}{l|}{585.63} & 654.25                  & \multicolumn{1}{|l|}{533.77} & \multicolumn{1}{l|}{609.42} & \multicolumn{1}{l|}{897.31}               \\ \hline
\multicolumn{1}{|c|}{4}                                    & \multicolumn{1}{l|}{794.06} & \multicolumn{1}{l|}{794.06} & 843.25                 & \multicolumn{1}{|l|}{1792.50} & \multicolumn{1}{l|}{1792.50} & \multicolumn{1}{l|}{1803.99}               \\ \hline
\multicolumn{1}{|c|}{5}                                   & \multicolumn{1}{l|}{475.41} & \multicolumn{1}{l|}{475.41} & 486.61                & \multicolumn{1}{|l|}{472.43} & \multicolumn{1}{l|}{472.43} & \multicolumn{1}{l|}{503.12}               \\ \hline

\multicolumn{1}{|c|}{6}                                    & \multicolumn{1}{l|}{1068.52} & \multicolumn{1}{l|}{1068.56} & 1076.81                  & \multicolumn{1}{|l|}{1271.35} & \multicolumn{1}{l|}{1271.32} & \multicolumn{1}{l|}{1648.49}               \\ \hline

\end{tabular}
\end{table}

\end{document}